\documentclass[a4paper,11pt]{article}
\pdfoutput=1 

\usepackage{jheppub} 

\usepackage[T1]{fontenc} 
\usepackage{amsthm}
\usepackage[T1]{fontenc}
\usepackage[utf8]{inputenc}
\usepackage{mathtools}   
\usepackage{amsfonts}
\usepackage{physics}
\usepackage{mathrsfs}

\theoremstyle{definition}
\newtheorem{definition}{Definition}[section]
\newtheorem{theorem}{Theorem}[section]
\newtheorem{corollary}{Corollary}[section]

\newtheorem{proposition}{Proposition}[section]
\newtheorem{example}{Example}[section]

\DeclareMathSymbol{\shortminus}{\mathbin}{AMSa}{"39}

\title{\boldmath Superconformal Quantum Mechanics on K\"ahler Cones}

\author{Nick Dorey,}
\author{Daniel Zhang}
\affiliation{Department of Applied Mathematics and Theoretical Physics, University of Cambridge\\Cambridge, CB3 0WA, UK}

\emailAdd{n.dorey@damtp.cam.ac.uk}
\emailAdd{d.zhang@damtp.cam.ac.uk}

\abstract{We consider supersymmetric quantum mechanics on a K\"{a}hler cone, regulated via a suitable resolution of the conical singularity. The unresolved space has a $\mathfrak{u}(1,1|2)$ superconformal symmetry and we propose the existence of an associated quantum mechanical theory with a discrete spectrum consisting of unitary, lowest weight representations of this algebra. We define a corresponding superconformal index and compute it for a wide range of examples. }

\arxivnumber{1911.06787}

\begin{document} 
\maketitle
\flushbottom

\section{Introduction}
Quantum mechanical models with an $SO(2,1)$ conformal symmetry \cite{deAlfaro}, and their superconformal extensions \cite{fubini}, are potentially of great interest due to their possible relation to $AdS_{2}$ holography. Concretely, many families of supersymmetric compactifications of string/M-theory to $AdS_{2}$ are known, see for example \cite{GKgeometry}, and it is natural to search for dual superconformal theories in one dimension. In this paper we will discuss the more basic question of formulating such theories together with the problem of defining and calculating suitable observables.\\

One of the simplest ways to achieve $SO(2,1)$ conformal invariance is to consider the motion of a particle on a manifold with a homothetic Killing vector or \textit{homothety} \cite{Michelson}. Here we will consider the case of a K\"{a}hler manifold with a holomorphic homothety or, equivalently, a {\em K\"{a}hler cone}. This builds on an earlier study of superconformal quantum mechanics on hyperK\"{a}hler cones by one of the authors \cite{doreybarns-graham}. Of course, any non-trivial conical geometry is singular and to define a sensible model it is necessary to work instead with a smooth resolution of the singular space. A conformal quantum mechanics might then be obtained by taking a suitable limit of the resolved space. In fact, there is at least one case where there are good reasons to believe that such a theory should exist. In \cite{ABS}, a conformally invariant quantum mechanics on the moduli space of Yang-Mills instantons (which is a hyperK\"{a}hler cone) was formulated and it was argued to provide a DLCQ formulation of the six-dimensional $(2,0)$-theory. In this context, the resolution of the conical singularity associated with small instantons was indeed interpreted as a UV regulator for the theory. In this paper we will adapt this viewpoint to a much wider class of models. Our main results are described in the rest of this introductory section. \\

Our starting point is a general K\"{a}hler cone $X$ with complex structure $I$, and holomorphic homothety\footnote{Equivalently, we have a cone over a \textit{Sasakian} space.} $D$. Any such space also has a canonical holomorphic isometry $D^{I}=I(D)$ associated with the \textit{Reeb vector field}. $X$ may also have additional holomorphic isometries with commuting generators $\{\mathcal{J}_{i}\}$. If so then, for a given complex structure $I$ on $X$, there could be infinitely many conical K\"{a}hler metrics corresponding to different choices of Reeb vector, which must lie in the $\textit{Reeb cone}$ as we shall describe later. This setting is familiar from the work of Martelli, Sparks and Yau (e.g. in \cite{sparksmartelliyauvolume}), who focused in particular on the case where $X$ obeys the Calabi-Yau condition and admits a unique Ricci-flat metric.\\

We begin by constructing an action of the superconformal algebra $\mathfrak{u}(1,1|2)$ on the space of differential forms, which is identified with the Hilbert space of supersymmetric quantum mechanics on $X$ in the usual way\footnote{Note however, as reviewed in section \ref{sectiongeometry} below, the inner product appropriate for superconformal quantum mechanics differs from the one for standard supersymmetric quantum mechanics. In particular, this difference is responsible for the existence of a discrete spectrum for the dilatation operator.} (this will require regularisation). The bosonic subalgebra is: 
\begin{equation}\label{Introbosonicsubalgebra}
\mathfrak{g}_{B} = \mathfrak{su}(1,1) \oplus \mathfrak{su}(2) \oplus \mathfrak{u}(1)_{R^I} \oplus \mathfrak{u}(1)_{D^I} 
\end{equation} 
where $\mathfrak{su}(1,1)\simeq\mathfrak{so}(2,1)$ is the conformal algebra, with Cartan generator $D$ realised geometrically as the Lie derivative with respect to the homothety. The $\mathfrak{su}(2)$ subalgebra is a nonabelian R-symmetry, with Cartan generator $J_3$ corresponding to the usual Lefschetz action on forms on a K\"{a}hler manifold. The factor $\mathfrak{u}(1)_{D^I}$ lies in the centre of the algebra and the generator $D^{I}$ corresponds to the Lie derivative with respect to the Reeb vector field. Importantly, as it is central in the superconformal algebra, $\mathfrak{u}(1)_{D^{I}}$ can mix with global symmetries. There is also an additional $\mathfrak{u}(1)_{R^I}$ factor with generator $R^{I}$ which is related to the difference $\frac{1}{2}(p-q)$ for forms of bidegree $(p,q)$. In the following the eigenvalues of the Cartan generators $\{D,J,D^{I},R^{I}\}$ of $\mathfrak{g}_{B}$ are denoted $\{\Delta,j,d,r\}$. Finally the algebra is completed by four supercharges $Q$ of positive dimension and four supercharges $S$ of negative dimension. \\

Our main hypothesis is that, for each K\"{a}hler cone $X$, there is an associated $\mathfrak{u}(1,1|2)$ superconformal quantum mechanics with a discrete spectrum consisting of unitary representations of this algebra.
As for superconformal algebras in higher dimensions, unitary representations of $\mathfrak{u}(1,1|2)$ can be classified according to their lowest weight $(\Delta,j,d,r)$. Lowest weights saturating the BPS bound $\Delta\geq 2j+d$ lead to "short" representations with (superconformal) primary states annhilated by some of the supercharges\footnote{Note that the supercharges $S$ annihilate all primary states.} Q. As the parameters of the theory vary, short representations can combine together to form long representations whose dimension can then be lifted above the bound. Following \cite{Romelsberger:2005eg, 4dscindex}, we can define a superconformal index which remains invariant (in a way to be made precise) under SUSY-preserving deformations of the theory. Picking conjugate supercharges $q$ and $s$ we define a Hamiltonian $\mathcal{H}$ as:
 \begin{equation}
\{q,s\} = \mathcal{H} = \frac{1}{2}\left(\mathbb{L}_0 + 2J_3-D^{I}\right) 
\end{equation} 
$\mathcal{H}$ has eigenvalues $E = \frac{1}{2}(\Delta - 2j - d)$. The superconformal index is then defined as:
\begin{equation}\label{Introscindex}
\mathcal{Z}_{X}(\tau,\tilde{y},Z) = \mbox{Tr} \left[(-1)^{F} e^{-\beta \mathcal{H}} \tau^{D^{I}} \tilde{y\textit{}}^{J_3+R^I} \prod_{i}z_{i}^{\mathcal{J}_i}\right]
\end{equation} 
The index is graded by the Cartan generators $\{D^{I},J_{3}+R^I\}$ of the $\mathfrak{u}(1|1)$ subalgebra of $\mathfrak{u}(1,1|2)$ which commutes with $q$ and $s$. The index is therefore a function of the corresponding fugacities $\tau$ and $\tilde{y}$ as well as the fugacities $\{z_{i}\}$ for the global symmetries $\{\mathcal{J}_{i}\}$. By construction, long representations of $ \mathfrak{u}(1,1|2)$ have vanishing contribution to the index while each short representation $R$ contributes a specific character of $\mathfrak{u}(1|1)$.\\

In order to properly define quantum mechanics on the K\"{a}hler cone $X$, we need to regulate the theory by considering a suitable resolution $\pi : \tilde{X}\rightarrow X$ of the singularity. If we want to define a regulated version of the index on the resulting smooth space then we need an equivariant resolution where the holomorphic isometries of $X$ are preserved on $\tilde{X}$. For \textit{regular} K\"{a}hler cones such a resolution has been shown to always exist by Martelli, Sparks and Yau \cite{sparksmartelliyauvolume}. There are also many other examples of cones (including those which are \textit{irregular} and \textit{quasi-regular}) of physical interest, such as Nakajima quiver varieties and toric varieties, for which equivariant resolutions have also been shown to exist.  On the resolved space one may then consider the trace corresponding to (\ref{Introscindex}), but now evaluated on forms on $\tilde{X}$. The resulting formula for the index is:
\begin{equation}\label{Introregularised index}
\mathcal{Z}_{\tilde{X}}(\tau,\tilde{y},Z) = \sum^{d_{\mathbb{C}}}_{p,q = 0} (-)^{p+q-d_{\mathbb{C}}}\, \tilde{y}^{p-d_{\mathbb{C}}/2}\, \mbox{Tr}_{H^{q}(\tilde{X}; A^{p}(\tilde{X}))}\left(\tau^{D^{I}}\prod_{i}z_i^{\mathcal{J}_i}\right) 
\end{equation} 
which coincides with the equivariant Hirzebruch $\chi_{s}$ genus (with $s=-\tilde{y}$) computed in the equivariant sheaf cohomology of $\tilde{X}$. As we discuss below, this invariant encodes both holomorphic and topological data of $\tilde{X}$. In particular, different limits of the index reduce to the Hilbert series counting holomorphic functions on the K\"{a}hler cone $X$, a series counting holomorphic sections of the canonical sheaf, and to the Poincar\'e polynomial of the preimage of the singularity.\\

Provided that the torus action $T$ associated with the holomorphic isometries $\{D^{I}, \mathcal{J}_{i}\}$ has a finite set $\tilde{X}^{T}$ of isolated fixed points, the index can then be computed by standard localisation theorems as: 
\begin{equation}\label{Introscindexlocalisationformula} 
\mathcal{Z}_{\tilde{X}}(\tau,\tilde{y},Z) =  \sum^{d_{\mathbb{C}}}_{p = 0} (-)^{p-d_{\mathbb{C}}}\, \tilde{y}^{p-d_{\mathbb{C}}/2}\sum_{x \in \tilde{X}^T} \mbox{ch}_T\left(\Lambda^{p}T^{*}_{x}, \tau, Z\right) \mbox{PE}\left[\mbox{ch}_T\left(T^{*}_{x}, \tau, Z\right)\right] 
\end{equation} 
where PE denotes a plethystic exponential and ch$_{T}$ denotes a character of the torus action evaluated on the tangent space to each fixed point.\\

Comparing the fixed-point formula (\ref{Introscindexlocalisationformula}) to the evaluation of the index (\ref{Introscindex}) on a generic spectrum of $\mathfrak{u}(1,1|2)$ representations yields detailed (but incomplete) information about the multiplicities of short and (semi-)short representations in each representation of the global symmetry. In particular, as in higher dimensional SCFTs, we find that there are certain protected multiplets which cannot be lifted. Comparison of the index evaluated on the spectrum with (\ref{Introregularised index}) shows that these are in one to one correspondence with holomorphic sections of the canonical bundle on the resolved space $\tilde{X}$. The geometric formula for the index also predicts the presence of "ground-state" representations whose primary states have $\Delta=2j$ and thus $d=0$. Interestingly these include special one-dimensional representations with $\Delta=j=d=0$ (but non-zero $r$), which are in particular singlets under the $SO(2,1)$ conformal group. The presence of such singlet representations is potentially important for $AdS_{2}$ holography. A singlet ground state is required in which to evaluate correlation functions analogous to those of higher dimensional CFTs corresponding to boundary insertions in Poincar\'e coordinates on $AdS$. \\

The index in the ground-state sector may be computed by taking the $\tau \rightarrow 0$ limit of the index. We show that the resulting ground-state index is equal (as a polynomial in $\tilde{y}$) to the Poincar\'e polynomial of the "core" of the resolved space. By this we mean the preimage of the singular point $\{o\}$ of $X$ under the resolution map; $\pi: \tilde{X}\rightarrow X$, i.e. $=\pi ^{-1}(\{o\})$. As this is a polynomial with positive coefficients, it provides a lower bound on the degeneracy of states saturating the bound $\Delta=2j$. Note however, that in general knowledge of the index alone is not sufficient to disentangle the singlets from the other ground-state representations. \\

Comparison between the algebraic and geometric formulae for the index, (\ref{Introscindex}) and (\ref{Introscindexlocalisationformula}), also yields some non-trivial consistency checks on our construction. First, for any K\"{a}hler cone, the superconformal index is locally independent of the resolution parameters. This is because it coincides with a holomorphic index on $\tilde{X}$. In the special case of toric Calabi-Yau 3-folds, we can also show that it is also invariant under wall crossing. For other cases, we show various limits of the index are invariant under choice of resolution. A particularly nice set of examples are Ricci-flat Calabi-Yau cones and their equivariant crepant resolutions. The superconformal index of these have an additional symmetry associated to the existence of a nowhere vanishing $(d_{\mathbb{C}},0)$ form, where $d_{\mathbb{C}}$ is the complex dimension of $X$. These results are consistent with our proposal of a $\mathfrak{u}(1,1|2)$ invariant theory associated to the underlying singular cone $X$.\\

Our study of the index also reveals interesting links to the work of Martelli, Sparks and Yau \cite{sparksmartelliyauvolume}. Variations of the Reeb vector correspond in superconformal quantum mechanics to mixing of the $\mathfrak{u}(1)_{D^I}$ R-symmetry with the global symmetries. This is implemented in the index precisely by a suitable rescaling of the fugacities. For the space of K\"ahler cones they consider, the $\tilde{y}\rightarrow 0$ of the index coincides with the Hilbert series of the singular K\"{a}hler cone $X$. Its resulting asymptotic behaviour in the limit of large charges is controlled by the volume of the corresponding Sasaki manifold. In the Calabi-Yau case, the unique Ricci flat metric is known to correspond to the stationary point of the volume. Thus the asymptotic growth of the index is maximised in the Calabi-Yau case.\\

In the body of the paper, we construct lowest weight, unitary, irreducible representations of $\mathfrak{u}(1,1|2)$, define the index and calculate it using equation (\ref{Introscindexlocalisationformula}) for a wide variety of examples in which the fixed point data is available. The case of toric cones is particularly tractable as the index can be expressed in terms of the toric data. We provide explicit calculations for low-dimensional examples such as the conifold and the $Y^{p,q}$ geometries. Special cases of particular interest include cones which satisfy the Calabi-Yau condition and admit a Ricci-flat metric. Another broad class of examples is provided by Nakajima quiver varieties and their $\mathbb{C}^{*}$ fixed subvarieties, such as the handsaw quiver variety of recent interest. We derive the Reeb cone explicitly for A-type quiver varieties.\\

Finally we note that models of the type we consider here arise in different physical contexts. First K\"{a}hler cones are ubiquitous as the Higgs branches of gauge theories with four supercharges and vanishing mass and FI parameters. The quantum mechanical $\sigma$-models of the type described above can thus arise via supersymmetry-preserving compactifications of higher-dimensional gauge theories. Indeed it is natural to conjecture that the $\mathfrak{u}(1,1|2)$ superconformal quantum mechanics described here is the endpoint of a corresponding RG flow across dimensions. A second, but related, context for these models is as the moduli-space quantum mechanics of solitons in scale-invariant gauge theories. In fact, in a future work with Samuel Crew \cite{crewdoreyzhang} we show that the vortex partition function of the 3d $\mathcal{N}=4$  $T_{\rho}(SU(N))$ gauge theory is generated by superconformal indices of quantum mechanics on handsaw quiver varieties, which are its vortex moduli spaces. Via these relations to SUSY gauge theory, the models considered here also have natural embeddings in string theory as the world-volume theories of D-branes. These are, in turn, a promising starting point for investigating possible holographic duals of superconformal quantum mechanics. We note however, that most of the supersymmetric $AdS_{2}$ geometries discussed in the recent literature correspond to quantum mechanical systems with $\mathcal{N}=(0,2)$ supersymmetry rather than the $\mathcal{N}=(2,2)$ case studied here. For this reason, it would be very interesting to extend our approach to investigate superconformal extensions of one-dimensional $\sigma$-models with $\mathcal{N}=(0,2)$ supersymmetry. We hope to return to these questions in a future publication.

\section{Superconformal Quantum Mechanics on K\"ahler Cones}
Following \cite{singletonexterioralgebra, doreysingleton, doreybarns-graham} we study the standard supersymmetric $\sigma$-model quantum mechanics on a Riemannian manifold $(X,g)$, with action given by:
\begin{equation}\label{nonlinearsigmamodel}
	S = \int dt \,\, \frac{1}{2}g_{\mu\nu}(X)\dot{X}^{\mu}\dot{X}^{\nu} + i g_{\mu\nu}(X)\psi^{\dagger\mu}\frac{D}{Dt}\psi^{\nu}-\frac{1}{4}R_{\mu\nu\rho\sigma}(X)\psi^{\dagger\mu}\psi^{\dagger\nu}\psi^{\rho}\psi^{\sigma}
\end{equation}
Here $\psi^\mu$ are sections of the odd cotangent bundle, and $\frac{D}{Dt} \psi^\mu= \nabla_{\dot{X}}\psi^\mu = \dot{\psi}^{\mu}+ \dot{X}^{\nu}\Gamma^{\mu}_{\nu\rho}\psi^{\rho}$ the induced covariant derivative. It is shown in \cite{singletonexterioralgebra} that this is invariant under $\mathcal{N} = (1,1)$ supersymmetry transformations, enlarged to $\mathcal{N} = (2,2)$ when $X$ is K\"ahler. We now specialise to this case. \\

Suppose that in addition $(X,g)$ admits a holomorphic closed \textit{homothety} $D$, i.e. a holomorphic vector field $D$ satisfying:
\begin{equation}\label{homothety}
	\mathcal{L}_{D}g=2g, \qquad \mathcal{L}_{D}K=2K, \qquad D_{\mu} = \partial_{\mu}K
\end{equation}
where K is the K\"ahler potential, one obtains $\mathfrak{so}(2,1)$ conformal algebra (generated by $\mathbb{D}$ and $\mathbb{K}$ associated to the vector field $D$ and scalar function $K$ respectively, and the Hamiltonian $\mathbb{H}$). Due to the holomorphy of $D$ and the fact that $K$ is the K\"ahler potential, this combines with the $\mathcal{N}=(2,2)$ algebra to give a $\mathfrak{u}(1,1|2)$ \textit{superconformal} algebra. For details, including a list of generators and the canonical quantisation with Hilbert space $\Omega^*(X,\mathbb{C})$ (the exterior algebra on $X$), see appendix \ref{appendixA}. For their full derivation and the transformations of the fundamental fields see \cite{andrewthesis}. \\

The conditions (\ref{homothety}) imply that $D^{\mu}_{\,\,;\nu} = \delta^{\mu}_{\,\,\nu}$. It was shown in \cite{gibbonsrychenkova} that this implies that $X$ is a cone\footnote{Whether or not we include the singularity $r=0$, denoted $\{o\}$, when discussing $X$ will be made clear from context in this work.} $C(Y)$ over a base manifold $Y$ (which is then by definition \textit{Sasaki}), i.e. that we can write:
\begin{equation}\label{kahlercone}
	g_{\mu\nu} dx^{\mu}dx^{\nu} = dr^2 + r^2 h_{ij}(\{x\})dx^{i}dx^{j} \qquad r \in \mathbb{R}_{+}
\end{equation}
where $\mu,\nu = 1,...,\mbox{dim}_{\mathbb{R}}(X)$ and $i,j = 1,...,\mbox{dim}_{\mathbb{R}}(X)-1$. In these coordinates:
\begin{equation}
	D = r\frac{\partial}{\partial r} \,\,,\quad K = \frac{1}{2} \,r^2
\end{equation}
In fact it is easy to see that a K\"ahler cone given by (\ref{kahlercone}), obeys all the conditions of (\ref{homothety}) with K\"ahler potential  $K = \frac{1}{2} \,r^2$ and hence defines an $\mathcal{N} = (2,2)$ superconformal quantum mechanics. It is necessary and sufficient then to look at K\"ahler cones. Such manifolds have a canonically defined holomorphic isometry generated by the \textit{Reeb vector}:
\begin{equation}
	D^{I}  = I(D)
\end{equation}
where $I$ is the complex structure on $X$.\\

 In the case when $Y$ is smooth and compact, it is known that $X=C(Y)$ is an affine variety \cite{orneaverbitsky}. The singularity $\{o\}$ in this case is isolated. In \cite{collinsthesis}, it is then shown that a K\"ahler cone $X$ over such a Sasakian base can be described as an affine scheme defined by ideals homogeneous under the action\footnote{This is a complex torus. The action of the complex torus is specified by that of the real torus action via multiplying induced vector fields by the complex structure $I$. Their fixed points coincide. Hereinafter we always refer to the action of the complexified torus, i.e. whenever there is a $U(1)$ action we always consider its complexification $\mathbb{C}^*$.} of a torus $T \equiv \mathbb{T}^s \equiv ({\mathbb{C}^*})^s$, such that the Reeb vector $D^{I} \in \mathfrak{\mathfrak{t}_s}$ (the Lie algebra of $T$) acts with positive weights on the non-constant holomorphic functions on $X$, and weight $0$ on the constant functions. The set of elements of $\mathfrak{t}_s$ which satisfy this define the $\textit{Reeb cone}$, which is a convex rational polyhedral cone as described in \cite{collinsthesis, sparksmartelliyauvolume}. We extend this definition of the Reeb cone to the case when $Y$ is potentially singular, but $X$ is still an affine scheme with a holomorphic torus action. In this case, there may be singular subspaces which intersect the tip $\{o\}$ of $X$. That is, we define the Reeb cone in this case to still be the subset of $\mathfrak{t}_s$ (the Lie algebra of a torus of isometries) under which the non-constant holomorphic functions are graded positively. Since the ring of global sections of the structure sheaf is finitely generated, the Reeb cone is also a convex rational polyhedral cone in this case \cite{collinsthesis}. In particular, this ensures that, choosing a Reeb vector in the Reeb cone, under the dilatation (which is the complexification of the Reeb vector) all points in $X$ contract to the tip of the cone.  The cases we consider where the Sasakian link $Y$ could be singular are quiver varieties. As we shall see, for these K\"ahler cones, the Reeb vector corresponding to the canonical metrics on them lie in the Reeb cone. \\

Our main hypothesis is that to each K\"ahler cone there is an associated $\mathfrak{u}(1,1|2)$ superconformal quantum mechanics. The predominant issue is that such spaces are singular (except when Y is the sphere). The space of forms arising in the canonical quantisation of the $\sigma$-model cannot be defined at singularities, and so such theories require regularisation. We will propose here that there is a quantity, the superconformal index, to which a regulated definition can be associated after resolving the singular cone by an equivariant resolution. The resolution preserves the algebra corresponding to the stabiliser of the BPS/unitary bound of the full superconformal quantum mechanics. We show that in many cases the index is independent of resolution and contains information about the spectrum of $\mathfrak{u}(1,1|2)$ multiplets on the cone.

\section{The Superconformal Index}
\subsection{Representation theory of $\mathfrak{u}(1,1|2)$}
The spectrum of the quantum mechanics should consist of a set of positive-energy, irreducible, unitary representations of the superconformal algebra $\mathfrak{u}(1,1|2)$, which we classify in this section in the usual way for Lie superalgebras. The bosonic subalgebra of the superconformal algebra is (again see appendix \ref{appendixA} for full details):
\begin{equation}\label{bosonicsubalgebra}
	\mathfrak{g}_{B} = \mathfrak{su}(1,1) \oplus \mathfrak{su}(2) \oplus \mathfrak{u}(1)_{R^I} \oplus \mathfrak{u}(1)_{D^I}
\end{equation}
with Cartan generators $\mathbb{D}$, $J_3$, $R^{I}$ and $D^{I}$. Note that the same letters are used to denote $D$ and $D^{I}$ as the vector fields that generate their action. For convenience we perform a change of basis (see \cite{andrewthesis}) of the $\mathfrak{su}(1,1) \cong \mathfrak{sl}(2,\mathbb{R})$ via:
\begin{equation}
	Z \mapsto e^{-\mu \mathbb{K}}e^{\frac{1}{2}\mu^{-1}\mathbb{H}} Z e^{-\frac{1}{2}\mu^{-1}\mathbb{H}}  e^{\mu \mathbb{K}} \qquad \quad \mu \in (0,\infty)
\end{equation}
For ease of notation we set $\mu =1$ but all of the following analysis can be performed for general $\mu$. We obtain:
\begin{eqnarray}
	\begin{split}
		i\mathbb{D} &\mapsto \mathbb{L}_{0} = \mu^{-1}\left(\mathbb{H}+\mu^2 \mathbb{K}\right)\\
		\mathbb{H} &\mapsto 2\mu\mathbb{L}_{-} = \mu\left(\mu^{-1}\mathbb{H}-\mu\mathbb{K}-i\mathbb{D}\right)\\
		\mathbb{K} &\mapsto -\frac{1}{2\mu} \mathbb{L}_{+} = -\frac{1}{4\mu} \left(\mu^{-1}\mathbb{H}-\mu \mathbb{K} + i\mathbb{D}\right)
	\end{split}
\end{eqnarray}
such that:
\begin{equation}
	\mathbb{L}_{0}^{\dagger} = \mathbb{L}_{0} \qquad \mathbb{L}_{+}^{\dagger} = \mathbb{L}_{-} \qquad [ \mathbb{L}_{0},  \mathbb{L}_{\pm}] =  \pm 2  \mathbb{L}_{\pm} \qquad  [ \mathbb{L}_{+},  \mathbb{L}_{-}] =  -   \mathbb{L}_{0} 
\end{equation}
The same change of basis is performed on the fermionic generators, and then linear combinations are taken so that the generators are eigenvalues of the adjoint action of the Cartans of the bosonic subalgebra. 
\begin{equation}
	\begin{split}
		q^{1+} &= \frac{1}{2}\left(Q^{\dagger}+iS^{\dagger}+i{Q^{I}}^{\dagger}-{S^{I}}^{\dagger}\right) \\
		q^{1-} &= \frac{1}{2}\left(Q^{\dagger}+iS^{\dagger}-i{Q^{I}}^{\dagger}+{S^{I}}^{\dagger}\right)  \\
		q^{2+} &= \frac{1}{2}\left(Q^{I}+iS^{I}-iQ+S\right) \\
		q^{2-} &= \frac{1}{2}\left(Q^{I}+iS^{I}+iQ-S\right) 
	\end{split}
	\qquad
	\begin{split}
		s^{1+} = {\left(q^{1+}\right)}^{\dagger} &= \frac{1}{2}\left(Q-iS-iQ^{I}-S^{I}\right) \\
		s^{1-} = {\left(q^{1-}\right)}^{\dagger} &= \frac{1}{2}\left(Q-iS+iQ^{I}+S^{I}\right) \\
		s^{2+} = {\left(q^{1-}\right)}^{\dagger} &= \frac{1}{2}\left({Q^{I}}^{\dagger} - i{S^{I}}^{\dagger}+iQ^{\dagger}+S^{\dagger}\right) \\
		s^{2-} = {\left(q^{1-}\right)}^{\dagger} &= \frac{1}{2}\left({Q^{I}}^{\dagger} - i{S^{I}}^{\dagger}-iQ^{\dagger}-S^{\dagger}\right)
	\end{split}
\end{equation}

These generators transform in the $(2\otimes2)_{+} \oplus(2\otimes2)_{-} $ of $\mathfrak{g}_{B}$ where the $\mathfrak{sl}(2,\mathbb{R})$ doublets are (with charges $(+1,-1)$ under $\mathbb{L}_{0}$ respectively):
\begin{equation}
	(q^{1+}, s^{2-})\qquad(q^{1-}, s^{2+})\qquad(q^{2+}, s^{1-})\qquad(q^{2-}, s^{1+})
\end{equation}
The $\mathfrak{su}(2)$ doublets are (with charges $(+1/2,-1/2)$ under $J_3$ respectively):
\begin{equation}
	(q^{2+}, q^{1+})\qquad(q^{2-}, q^{1-})\qquad(s^{1+}, s^{2+})\qquad(s^{1-}, s^{2-})
\end{equation}
Also $q^{a\pm}$ has $U(1)_{R^I}$ charge $\pm\frac{1}2$ and $s^{a\pm}$ charge $\mp\frac{1}2$. They obey commutation relations:
\begin{equation} 
	\begin{split}
		\left\{ q^{ai}, s^{bj}\right\} &=  \delta^{ij}\delta^{ab}\mathbb{L}_{0}+2\delta^{ij}\mathcal{J}^{ab}-\delta^{ab}{\sigma_{3}}^{ij}D^{I}\\
		\left\{ q^{ai}, q^{bj}\right\} &= -2\epsilon^{ij}{\sigma_2}^{ab}\mathbb{L}_{+}\\
		\left\{ s^{ai}, s^{bj}\right\} &= +2\epsilon^{ij}{\sigma_2}^{ab}\mathbb{L}_{-}
	\end{split}
\end{equation}
Where $a,b \in \{1,2\}$,  $i, j \in \{+,-\}$,
$ \mathcal{J}^{ab} = \left( \begin{array}{cc}
J_3 & J_- \\
J_+ & \shortminus J_3 
\end{array} \right)$ and $ \epsilon^{ij} = \left( \begin{array}{cc}
0 & 1 \\
\shortminus 1 & 0 
\end{array} \right)$.\\

Unitary irreducible representations of $\mathfrak{g}_{B}$ are labelled by eigenvalues of the Cartans on the lowest weight state. The presence of the abelian summands does not affect this analysis since \textit{unitary} irreducible representations of these act by constant multiplication - unitary representations of $\mathfrak{u}(1)$ are represented by hermitian or anti-hermitian operators depending on convention, which are always diagonalisable. Since all other generators in (\ref{bosonicsubalgebra}) commute with a given $\mathfrak{u}(1)$ generator, irreducible representations of $\mathfrak{g}_B$ are labelled by a single eigenvalue for each $\mathfrak{u}(1)$.\\

Unitary irreducible representations of the full superconformal algebra $\mathfrak{u}(1,1|2)$ are in particular unitary representations of $\mathfrak{g}_{B}$ and therefore direct sums of the irreducible representations of $\mathfrak{g}_{B}$ mentioned above. In particular we can restrict ourselves to considering lowest weight irreducible representations of $\mathfrak{u}(1,1|2)$ corresponding to lowest weights of the bosonic subalgebra, because the lowering operators $s^{a\pm}$ lower the $\mathbb{L}_0$ eigenvalue. Such lowest weight unitary representations must contain the states formed from the action of $\mathfrak{g}$ inherited from the action on the full Hilbert space. Following \cite{andrewthesis} one can show that if the Verma module specified by a lowest weight state (with unit norm under a candidate inner product) contains no negative-norm states, then the maximal proper submodule consists of zero-norm states and can be quotiented out to give an unitary irreducible representation.\\

We therefore work with lowest weight representations, with superconformal primary state $v_{\lambda} \equiv \ket{\Delta, j, r, d}$ such that:
\begin{equation}
	\begin{split}
		\mathbb{L}_{0}\ket{\Delta, j, r, d}&= \Delta\ket{\Delta, j, r, d} \\
		J_{3}\ket{\Delta, j, r, d} &= -j \ket{\Delta, j, r, d}
	\end{split}
	\qquad\qquad
	\begin{split}
		R^I\ket{\Delta, j, r, d}&= r\ket{\Delta, j, r, d}  \\
		D^{I}\ket{\Delta, j, r, d} &= d \ket{\Delta, j, r, d}
	\end{split}
\end{equation} 
By definition $\ket{\Delta, j, r, d}$ is annihilated by all lowering operators of the algebra $\mathfrak{g}_{-}$ which we choose to be spanned by $\{s^{a\pm}, \mathbb{L_{-}}, J_{-}\}$:
\begin{equation}
	\mathfrak{g}_{-} \ket{\Delta, j, r, d} = 0
\end{equation}
Such a state exists since $\{s^{a\pm}\}$ lower the eigenvalue of $\mathbb{L}_{0}$, which is bounded below for all states in the theory. We choose $\Delta \geq 0$, $2j \in \mathbb{N}$ to ensure the module generated by the action of the bosonic subalgebra is unitary and irreducible, and we restrict to the case where $r \in \mathbb{Z}/2$ due to the fact that the generator $R^I$ descends from a full $U(1)_{R^I}$ group action, and its normalisation in the algebra.\\

In general however $d$ can be any real number, as it need not descend from a full $U(1)$ group action. $D^{I}$ corresponds to flow along the Reeb vector, whose orbits can either close, or not. If they all close then the Reeb vector induces a full $U(1)$ action on the K\"ahler cone $X$ (in fact the flow is solely along the Sasakian link Y), which is either locally free or free. Sasaki metrics corresponding to these cases are referred to as \textit{quasi-regular} and \textit{regular} respectively, and when $X$ is the metric cone over such a Sasaki metric the eigenvalues of $D^I$ will be integer-valued when its action is appropriately normalised, or more generally integer multiples of a constant. When the orbits of $D^{I}$ do not close, the Sasaki manifold is said to be \textit{irregular}, and the eigenvalues $d$ are unconstrained (although in all three cases the eigenvalues of $d$ are non-negative as we shall see later). We will assume that the spectrum of the quantum mechanics is discrete and therefore that the index we later define to still be well-defined in this case.\\

Note that when $X$ is a hyperK\"ahler cone the holomorphic isometry $D^{I}$ is given by a linear combination of Cartan generators lying in non-abelian subalgebras of the superconformal algebra \cite{doreybarns-graham}, and therefore its eigenvalues must be quantised. This corresponds to the fact that all hyperK\"ahler cones, whose link are by definition \textit{3-Sasakian}, are regular or quasi-regular when considered as K\"ahler cones due to the non-abelian $\mathfrak{su}(2)$ generated by the triplet of Reeb vector fields with respect to each complex structure. \\ 

One can show \cite{andrewthesis,dobrevpetkova} that necessary and sufficient conditions for unitarity of a given lowest weight representation are:
\begin{equation}
	|| \,({q^{1+}})^{n_1}({q^{1-}})^{n_2}({q^{2+}})^{n_3}({q^{2-}})^{n_4} \ket{\Delta, j, r, d} \,||^{2} \geq 0 \qquad\quad \forall  \, n_{i} \in \{0,1\}
\end{equation}
By explicit computation, the most stringent bounds occur at level 1, and are:
\begin{equation}\label{unitarybounds}
	\begin{split}
		|| \,{q^{1+}}\ket{\Delta, j, r, d}\,||^{2} &= \Delta - 2j - d\geq 0 \\
		|| \,{q^{1-}}\ket{\Delta, j, r, d} \,||^{2} &= \Delta - 2j + d\geq 0 
	\end{split}
	\qquad\qquad
	\begin{split}
		|| \,{q^{2+}}\ket{\Delta, j, r, d} \,||^{2} &= \Delta + 2j - d\geq 0   \\
		|| \,{q^{2-}}\ket{\Delta, j, r, d} \,||^{2} &= \Delta + 2j + d\geq 0 
	\end{split}
\end{equation} 
If there are no negative norm states then a unitary irreducible representation is obtained by quotienting out zero norm states. We thus obtain a relationship between conditions for representation shortening and the existence of BPS states: lowest weight states annihilated by 1 or more supercharges. We now restrict to the case $d \geq 0$, noting that an analogous analysis can be made in the $d \leq 0$ case. Our index will only receive contributions from states with $d \geq 0$. We will see that the $d \geq 0$ index corresponds to the $\bar{\partial}$-cohomology on $X$ (strictly $\tilde{X}$, see later). Performing the analogous computations for $d \leq0$ and defining the corresponding index can then easily be seen to correspond to the $\partial$-cohomology. 
\begin{proposition}\label{reptypes}
	The lowest weight unitary irreducible representations of $\mathfrak{u}(1,1|2)$ are of the following type:
	\begin{itemize}
		\item Long representations $L(\Delta, j, d, r)$: $\Delta > 2j + d$.
		\item $\frac{1}{4}$-BPS short representations $S_{1/4}(j,d,r)$: $j \neq 0$, $d\neq0$ and $\Delta = 2j+d$. Here $q^{1+}v_\lambda$ has zero norm and is quotiented out, or equivalently annihilated. 
		\item $\frac{1}{2}$-BPS short representations $S_{1/2}(j,r)$: $j\neq 0$, $d=0$ and $\Delta = 2j$. Here both $q^{1+}$ and $q^{1-}$ annihilate the ground state. 
		\item Special $\frac{1}{2}$-BPS short representations $S'_{1/2}(d,r)$: $j=0$, $d\neq0$ and $\Delta = d$. Here both $q^{1+}v_{\lambda}$ and $q^{2+}v_{\lambda}$ have zero norm but are not independent in the Verma module, since $J_{+}q^{1+}v_{\lambda} = -q^{2+}v_{\lambda}$ (as $J_+ v_{\lambda} = 0$ when $j=0$).
		\item Special maximally BPS short representations $S'_{1}(r)$: $j=d=0$, $r\neq0$. These are not the vacuum representation since $r\neq0$. Here all supercharges annihilate the lowest weight state, and so do all bosonic raising operators. Therefore we just obtain a rep of the $\mathfrak{u}(1)_{R^I}$.
	\end{itemize}
\end{proposition}
Note that the $S'_{1}$ representations are singlets, and in particular are invariant under the $\mathfrak{so}(2,1)$ conformal algebra. They are therefore candidate ground states in the $CFT_1$ of an $AdS_2/CFT_1$ duality. \\

We construct an index: a count of short representations which is (in a way to be made precise) invariant under certain deformations of the theory. In order for such an index to be a count of short representations it must be invariant under the situation in which, for a long representation, the quantity $\epsilon = \Delta -2j -d$ (assuming $d\geq 0$) continuously lowers to $0$ and the unitary bound is reached. The long representation splits into a direct sum of short representations containing fewer states. Note that at most $\Delta$ and $d$ may vary continuously since $j$ and $r$ are quantised. Of course this process can also happen in reverse, where two short representations pair into a long representation which then moves away from the unitary bound. Any index which counts short representations must be invariant under these processes.\footnote{Note that when $d$ is quantised, necessarily the case for regular and quasi-regular cones, the short representations are protected under continuous deformations of the theory preserving the regularity/quasi-regularity property i.e. the fact that $D^I$ generates a full $U(1)$ action. This is because their dimension is related to their (quantised) R-charges. } There are 4 ways in which this can happen (when we have not specified that $j,d,r$ is 0 below, we mean that it is non-zero):
\begin{equation}\label{repsplitting}
	\begin{split}
		L|_{\epsilon = 0}(\Delta=2j+d, j, d, r) \,\,\,\,\,\quad &= S_{1/4}(j,d,r) \oplus S_{1/4}(j+1/2,d,r+1/2)\\
		L|_{\epsilon = 0}(\Delta=2j, j, d = 0, r)\,\,\,\,\, \quad&=  \,S_{1/2}(j,r) \oplus S_{1/2}(j+1/2,r+1/2) \\ &\qquad\oplus S_{1/2}(j+1/2,r-1/2) \oplus S_{1/2}(j+1,r)\\
		L|_{\epsilon = 0}(\Delta= d, j = 0, d, r) \qquad \,&= S_{1/2}'(d,r) \oplus S_{1/4}(1/2,d,r+1/2)\\
	    L|_{\epsilon = 0}(\Delta=0, j=0, d = 0, r) \,\,&=  \,S'_{1}(r) \oplus S_{1/2}(1/2,r+1/2) \\ &\qquad\oplus S_{1/2}(1/2,r-1/2) \oplus S_{1/2}(1,r)
	\end{split}
\end{equation}
Justifying this case by case:
\begin{itemize}
	\item $\mathbf{j \neq 0, d\neq 0}$: $q^{1+} v_{\lambda}$ becomes null and splits off into a $\frac{1}{4}$-BPS representation with lowest weight $(\Delta+1, j+\frac{1}{2}, d, r+\frac{1}{2}) $. 
	\item $\mathbf{j \neq 0, d=0}$: $q^{1+}v_{\lambda}$ and $q^{1-}v_{\lambda}$ become null, and also saturate the $1/2$ BPS bound. They are independent (one cannot be obtained under the action of the algebra from the other). If we assume that they become lowest weight states of irreducible null representations, then we might worry that they also contain null states. $q^{1+}$ and $q^{1-}$ are both nilpotent, but note that both lowest weight irreps would contain the state $q^{1+}q^{1-}v_{\lambda} = -q^{1-}q^{1+}v_{\lambda}$ which is null. This forms the lowest weight state of another null representation. 
	\item $\mathbf{j=0, d\neq0}$: $q^{1+}v_{\lambda}$ and $q^{2+}v_{\lambda}$ become null but are not independent, hence we obtain a null $1/4$ BPS representation (by checking the quantum numbers) with lowest weight vector $q^{1+}v_{\lambda}$. 
	\item $\mathbf{j=d=0}$: $q^{1\pm}v_{\lambda}$ and $q^{2\pm}v_{\lambda}$ become null but note the latter are obtained from the former via the action of $J_{+}$ and are hence not independent. We have a similar situation to $j\neq0$, $d=0$, with the null representation content being irreducible representations with lowest weight vectors $q^{1\pm}v_{\lambda}$ and $q^{1+}q^{1-}v_{\lambda}$. 
\end{itemize}
Note that this is conjectural, as we do not have a proof that the representations which become null are themselves irreducible. \\

A count of short representations invariant under the representation splitting is of the form:
\begin{equation}
	I = \sum_{R\in \mathcal{R}} \alpha(R)N(R)
\end{equation}
where $N(R)$ is the number of representations of type R present in the spectrum of the theory, $\mathcal{R}$ the set of possible short representations and $\alpha(R)$ a set of coefficients, which by (\ref{repsplitting}) must satisfy:
\begin{equation}\label{indexconditions}
\openup 1\jot
	\begin{split}
		0 &= \alpha \left(S_{1/4}(j,d,r)\right)+ \alpha\left(S_{1/4}\left(j+1/2,d,r+1/2\right)\right) \qquad\quad\,\,\,\, j\neq0, d\neq0 \\
		0 &= \alpha\left(S_{1/2}(j,r)\right) + \alpha\left(S_{1/2}(j+1/2,r+1/2)\right) \\ &\qquad+ \alpha\left(S_{1/2}(j+1/2,r-1/2)\right) + \alpha\left(S_{1/2}(j+1,r)\right) \qquad j\neq0, d=0\\
		0 &= \alpha \left(S'_{1/2}(d,r)\right)+ \alpha\left(S_{1/4}\left(1/2,d,r+1/2\right)\right) \qquad\qquad\quad\,\,\,\,\, j=0, d\neq0\\
		0 &= \alpha\left(S'_{1}(r)\right) + \alpha\left(S_{1/2}(1/2,r+1/2)\right) \\ &\qquad\qquad+ \alpha\left(S_{1/2}(1/2,r-1/2)\right) + \alpha\left(S_{1/2}(1,r)\right) \qquad\quad\,\, j=0, d=0
	\end{split}
\end{equation} 
Solving the constraints gives the following basis of indices:
\begin{equation}\label{indexbasis}
	I^{d,r}= \sum^{\infty}_{j\in \mathbb{N}_{0}/2} (-1)^{2j} N(S(j,d,r+j))
\end{equation}
Where $N(S(j,d,r))$ is the number of representations present of type:
\begin{equation}
	S(j,d,r) =
	\begin{cases} 
		S_{1/4}(j,d,r) &\mbox{if } j\neq0, d\neq 0  \\ 
		S_{1/2}(j,r) &\mbox{if } j\neq0, d= 0 \\
		S'_{1/2}(d,r) &\mbox{if } j=0, d\neq0 \\
		S'_{1}(r) &\mbox{if } j=0, d=0
	\end{cases}
\end{equation}

\subsection{The Superconformal Index}
We now seek to define the superconformal index, originally introduced for 4d field theory in \cite{Romelsberger:2005eg, 4dscindex}, which receives contributions solely from the short representations which it counts up to the splitting (\ref{repsplitting}). We choose a supercharge $q$ with hermitian conjugate $s=q^{\dagger}$ such that:
\begin{equation}
	\{q,s\} = \mathcal{H} = \frac{1}{2}\left(\mathbb{L}_0 + 2J_3-D^I\right)
\end{equation}
$\mathcal{H}$ has eigenvalues $E = \frac{1}{2}(\Delta - 2j - d)$ which coincides with the unitary bound. It is clear that the correct choice is $q=q^{1+}$. Each short multiplet contains states which are annihilated by $\mathcal{H}$, and long multiplets contain none. These states are in bijection with the cohomology classes of $s$ (or $q$) provided the spectrum is discrete. \\

The choice of supercharge breaks the full $\mathfrak{u}(1,1|2)$ to the subalgebra spanned by generators (anti)commuting with $q$ and $s$. This is the little group (algebra) and is denoted:
\begin{equation}
	\mathfrak{g}_0 = \langle \mathcal{H},q,s, J_3  +R^I, D^{I}, q^{1-}, s^{1-} \rangle
\end{equation}
Note that $I = \langle \mathcal{H},q,s \rangle$ is an ideal, such that:
\begin{equation}
	\frac{\mathfrak{g_0}}{I} \cong \mathfrak{u}(1|1)
\end{equation}
and we will henceforth refer to the above as the little group. The $\mathfrak{u}(1|1)$ is generated by $\{J_3  +R^I, D^{I}, q^{1-}, s^{1-}\}$ with the former two elements generators of its Cartan subalgebra.\footnote{Strictly speaking $\mathfrak{u}(1|1)$ is generated by the elements $(J_3 + R)+ I$ etc, but the abuse of notation is inconsequential since all elements of $I$ evaluate to 0 on the states which contribute to the index.} \\

We are now ready to formulate the superconformal index as: 
\begin{equation}\label{scindex}
	\mathcal{Z}_{X}(\tau,\tilde{y},Z) = \mbox{Tr} \left[(-1)^{F} e^{-\beta \mathcal{H}} \tau^{D^{I}} \tilde{y}^{J_3+R^I} \prod_{i}z_{i}^{\mathcal{J}_i}\right]
\end{equation}
where $F= 2J_3$, and $\{\mathcal{J}_i\}$ are a set of any additional mutually commuting global symmetry generators (which we restrict to be holomorphic isometries), graded with fugacities $Z \equiv \{z_i\}$. Henceforth we choose $\{\mathcal{J}_i\}$ to generate the remainder of the algebra of the torus $\mathfrak{t}_s$ defined previously. Note that relabelling $\tilde{y} = y/\tau$ we recover the superconformal index for hyperK\"ahler cones as in \cite{doreysingleton}, and indeed the little group of $\mathfrak{u}(1,1|2)$ indeed coincides with that of $\mathfrak{osp}(4^*|4)$ when $\left(X,g\right)$ is hyperK\"ahler. \\

By standard arguments as for the Witten index \cite{wittenconstraintsonsusybreaking}, assuming a discrete spectrum, the index is independent of $\beta$ and only receives contributions from $E=0$ states. States are matched in boson/fermion pairs with the same quantum numbers for $E>0$. Under continuous deformations of the theory preserving $\mathfrak{g}_0$, states with $E>0$ can lower to $E=0$, and states with $E=0$ to can lift to $E>0$, but can only do so in pairs with the same quantum numbers. $J_3$ and $R$ have quantised eigenvalues, and therefore the eigenvalues of states under these operators will not vary under continuous deformation. In general however the eigenvalues under $D^I$ may vary continuously. This is in line with expectation, we do not necessarily expect that superconformal quantum mechanics on different K\"ahler cones will yield the same superconformal index when graded by the Reeb vector, which partly specifies the K\"ahler cone structure.\footnote{In \cite{doreybarns-graham} only deformations to spaces on which the generator of the holomorphic isometry exponentiates to a full $U(1)$ action are considered.} If we do not grade by $D^I$ (i.e. setting $\tau =1$), then the index would be invariant under arbitrary continuous deformations preserving $\mathfrak{g}_0$. Grading by $D^I$, although the states which do not cancel and therefore contribute to the index (\ref{scindex}) track through the continuous deformation, their $D^I$ eigenvalue may differ and therefore the superconformal index will be invariant only up to the $\tau$ dependence of its expansion in terms of $\mathfrak{u}(1|1)$ characters. Later we show that the index receives contributions only from short representations, and that it is invariant under the representation splitting (\ref{repsplitting}).\\

On a given K\"ahler cone, it may be possible to specify multiple K\"ahler metrics with fixed complex structure $I$. If $D^{I} \in \mathfrak{t}_s$ where $t_{s} = \mathcal{L}(\mathbb{T}^s)$ and $\mathbb{T}^s$ is an $s$-torus of holomorphic isometries, we say the torus action is of \textit{Reeb-type}. Henceforth unless specified otherwise this is assumed for the cones considered in this work.\footnote{Those cones which are not of Reeb type are enumerated in \cite{lermanconctacttoricmanifolds}.} Superconformal indices corresponding to different K\"ahler metrics on a given cone are therefore given simply by a relabelling of fugacities in (\ref{scindex}). In the volume minimisation considered in \cite{sparksmartelliyauvolume}, a space of K\"ahler metrics is considered via varying the Reeb vector in order to obtain a Ricci-flat metric on the cone. Given any reference metric in this space, the superconformal indices of all other metrics in the space can be obtained by the fugacity relabelling mentioned above. \\
	
The $E=0$ states in the short representations transform in representations of $\mathfrak{u}(1|1)$ and contribute to the index via their character. These are:
\begin{equation}\label{littlegroupcharacter}
	\chi_R (\tau, \tilde{y}) =
	\begin{cases} 
		0 &\quad \phantom{\mbox{if}}\qquad  R = L(\Delta, j, d, r) \\
		\tau^{d}\,\tilde{y}^{-j+r}\left(1-\frac{1}{\tilde{y}}\right) &\quad \phantom{\mbox{if}}\qquad R =  S_{1/4}(j,d,r) \quad j\neq 0, \, d\neq 0\\
		\tilde{y}^{-j+r} &\quad \mbox{if}\qquad R = S_{1/2}(j,r) \quad j\neq 0\\
		\tau^{d}\,\tilde{y}^{r}\left(1-\frac{1}{\tilde{y}}\right) &\quad \phantom{\mbox{if}}\qquad R = S'_{1/2}(d,r) \quad d\neq 0 \\
		\tilde{y}^{r} &\quad \phantom{\mbox{if}}\qquad  R = S'_{1}(r)
	\end{cases}
\end{equation}

The $(d\geq0)$ short spectrum of a $\mathfrak{u}(1,1|2)$ invariant quantum mechanics  can be written:
\begin{equation}
	\begin{split}
		\mathcal{S} &= \left( \bigoplus_{d>0, \, j\in \mathbb{N}_{+}/2, \,r\in \mathbb{Z}/2} N(S_{1/4}(j,d,r))S_{1/4}(j,d,r) \right) \bigoplus \left( \bigoplus_{j\in \mathbb{N}_{+}/2,\, r\in \mathbb{Z}/2} N(S_{1/2}(j,r))S_{1/2}(j,r)\right)\\
		&\qquad \bigoplus \left( \bigoplus_{d>0,\, r\in \mathbb{Z}/2}N(S'_{1/2}(d,r)) S'_{1/2}(d,r) \right) \bigoplus \left( \bigoplus_{r\in \mathbb{Z}/2} N(S'_{1}(r))S'_{1}(r)\right)
	\end{split}
\end{equation}
where $N(R)$ are positive integers. Note the sum over $d$ is over $\mathbb{Z}_{+}$ multiples of a constant in the case where $\left(X,g\right)$ is quasi-regular/regular, and over arbitrary values of $\mathbb{R}_{+}$ when it is irregular. The superconformal index can be expressed (now allowing $N(R)(Z) \in \mathbb{Z}_{\geq 0 }\left[Z,Z^{-1}\right]$ to allow for grading by global symmetries, i.e. these are characters of the additional mutually commuting global symmetries):
\begin{equation}
	\begin{split}
		\mathcal{Z}_{X}(\tau,\tilde{y},Z)&=  \,\,\,\sum_{d>0,\, j\in \mathbb{N}_{+}/2, \,r\in \mathbb{Z}/2}(-1)^{2j}\tau^{d}\,\tilde{y}^{-j+r}\left(1-\frac{1}{\tilde{y}}\right) N(S_{1/4}(j,d,r))(Z)\\
		&\quad+ \sum_{j\in \mathbb{N}_{+}/2,\, r\in \mathbb{Z}/2}(-1)^{2j}\tilde{y}^{-j+r} N(S_{1/2}(j,r))(Z)\\
		&\quad+ \quad \sum_{d>0,\, r\in \mathbb{Z}/2}\tau^{d}\,\tilde{y}^{r}\left(1-\frac{1}{\tilde{y}}\right)N(S'_{1/2}(d,r))(Z)\\
		&\quad+\,\,\quad \sum_{r\in \mathbb{Z}/2}\tilde{y}^{r} N(S'_{1}(r))(Z)
	\end{split}
\end{equation}
Note that the index receives contributions from short representations where $d\geq0$ only, since from (\ref{unitarybounds}) for $d\leq0$ the appropriate unitary bound restriction on the lowest weight state is $\Delta-2j+d \geq 0$, thus the eigenvalue of $\mathcal{H}$ on the lowest weight state is $E=\Delta-2j-d >0$. Since all raising operators have $E$-grade $\geq 0$,  any unitary irreducible lowest weight representation with $d<0$ contains only $E>0$ states and therefore does not contribute to the index.  \\

The index can be expressed: 
\begin{equation}\label{scindexexpansion}
	\mathcal{Z}_{X}(\tau,\tilde{y} ,Z) = \sum_{r\in \mathbb{Z}/2}\tilde{y}^{r}I^{0,r} + \sum_{d > 0, \,r\in \mathbb{Z}/2}\tau^{d}\,\tilde{y}^{r}\left(1-\frac{1}{\tilde{y}}\right)I^{d,r}
\end{equation}
where:
\begin{equation}
	I^{d,r}= \sum^{\infty}_{j\in \mathbb{N}_{0}/2} (-1)^{2j} N(S(j,d,r+j))(Z)
\end{equation}
Setting $Z=1$ these coincide with (\ref{indexbasis}), so encouragingly the superconformal index (\ref{scindex}) may be written in a basis of indices invariant under representation splitting. Given $\mathcal{Z}_{X}(\tau,\tilde{y} ,Z)$, and using geometric constraints on the bidegree of forms in the Hilbert space,\footnote{Strictly speaking these constraints apply only to the resolution of $\left(X,g\right)$ which we will define later, but we assume these constraints also hold for states in the Hilbert space of the full $\mathfrak{u}(1,1|2)$ quantum mechanics i.e. that the procedure for removing the regulator is sufficiently 'smooth', so that since the $\mathfrak{u}(1)_{R^I}$ eigenvalues are quantised their bounds should not change under the deformation.} it is possible to read off $\{I^{d,r}\}$. This provides information on lower bounds of degeneracies of superconformal multiplets. Note that the $\tau = 0$ sector of the index yields information about the singlet representations/states $S_1'(r)$, interesting in any $AdS_{2}/CFT_{1}$ application.\\

 In special cases it is possible to uniquely determine the degeneracies of certain superconformal multiplets. In particular, this is the case when $r=d_{\mathbb{C}}/2$, where $d_{\mathbb{C}}$ is the complex dimension of $X$. This is because $N(j,d, r>d_{\mathbb{C}}/2) = 0 $ via geometric constraints on the bidegrees of forms and:
\begin{equation}
	I^{d,d_{\mathbb{C}}/2} = N\left(S\left(0,d, d_{\mathbb{C}}/2\right)\right)(Z) \qquad d \geq 0
\end{equation}
We call these multiplets \textit{protected}, and they are: 
\begin{equation}\label{protectedreps}
	\begin{split}
		&S'_{\frac{1}{2}}\left(d,r=d_{\mathbb{C}}/2\right) \qquad d\neq0\\
		&S'_1\left(r=d_\mathbb{C}/2\right)
	\end{split}
\end{equation}

\subsubsection{Geometric Interpretation of the Index}\label{sectiongeometry}
In \cite{singletonexterioralgebra} Singleton constructed a geometric action of $\mathfrak{u}(1,1|2)$ on the exterior algebra of a K\"ahler cone $X$. Strictly speaking this is only rigorously defined on flat space, but will suggest a regularised definition of the index for general K\"ahler cones. The supercharge $s$ acts on a form $\alpha = \beta e^{-K}$ as:
\begin{equation}
	s\alpha = \bar{\partial}\beta e^{-K}
\end{equation}
and so as the Dolbeault operator up to the exponential factor. Note that the $L^2$ cohomology of $s$ with respect to the usual inner product on $\Omega^*(X, \mathbb{C})$ is isomorphic to the usual $\bar{\partial}$ cohomology acting on forms which are $L^2$ with respect to the inner product:
\begin{equation}\label{innerproduct}
	(\alpha,\beta) = \int_{X} d^{2 d_{\mathbb{C}}}x \sqrt{g} \alpha \wedge \bar{\beta} e^{- K}
\end{equation}
where $K$ is the K\"ahler potential. Viewing the K\"ahler space $X$ as a holomorphic manifold, the R-symmetry generated by $D^{I}$  corresponds to a holomorphic $\mathbb{C}^{\times}$ action on $X$.\\

The resulting index is a trace over the space of states with finite norm under (\ref{innerproduct}) and vanishing $\mathcal{H}$ eigenvalue, graded by the two Cartan elements of the little group $\mathfrak{u}(1|1)$ and any holomorphic isometries of the manifold. In the case of affine space $X = \mathbb{C}^n$, the space of forms with finite norm under (\ref{innerproduct}) is isomorphic to $\mathbb{C}\left[z_i, \bar{z}_i, dz_i, d\bar{z}_i\right]$, and the space of states with $E=0$ is isomorphic to  $\mathbb{C}\left[z_i, dz_i\right]$ i.e. the polynomial-valued holomorphic forms. We compute and analyse the index fully in appendix \ref{appendixB}. It is then natural to work in the basis of homogeneous polynomials. In the flat space case then, the analytic $\bar{\partial}$-cohomology on polynomial valued forms coincides with the sheaf cohomology of $\mathbb{C}^n$ considered as an affine variety. We will assume this holds true for a general K\"ahler cone. Therefore the index (\ref{scindex}) can be expressed:

\begin{equation}\label{scindexcharacterform}
	\begin{split}
		\mathcal{Z}_{X}(\tau,\tilde{y},Z) &= \mbox{Tr}_{\Omega^{*}(X)} \left[(-1)^{F} e^{-\beta \mathcal{H}} \tau^{D^{I}} \tilde{y}^{J_3+R^I} \prod_{i}z_{i}^{\mathcal{J}_i}\right]\\
		&= \sum^{d_{\mathbb{C}}}_{p,q = 0} (-)^{p+q-d_{\mathbb{C}}}\, \tilde{y}^{p-d_{\mathbb{C}}/2} \mbox{Tr}_{H^{p,q}(X)}\left(\tau^{D^{I}}\prod_{i}z_i^{\mathcal{J}_i}\right)\\
		&= \sum^{d_{\mathbb{C}}}_{p,q = 0} (-)^{p+q-d_{\mathbb{C}}}\, \tilde{y}^{p-d_{\mathbb{C}}/2} \mbox{Tr}_{H^{q}(X; A^{p}(X))}\left(\tau^{D^{I}}\prod_{i}z_i^{\mathcal{J}_i}\right)
	\end{split}
\end{equation}
In the last line we have used Dolbeault's theorem, which states that for $\mathcal{M}$ a complex manifold: $H^{q}(\mathcal{M}; A^{p}(\mathcal{M}))= H^{p,q}(\mathcal{M})$ where the left hand side is the sheaf cohomology of $A^{p}(\mathcal{M})$ -  the sheaf of holomorphic $(p,0)$ forms on $\mathcal{M}$.\\

The lowest weight states of the protected representations (\ref{protectedreps}) are then in bijection with the holomorphic forms of bidegree $(p,q) =(d_{\mathbb{C}},0)$, i.e. holomorphic sections of the canonical bundle.\\

Note that strictly speaking all of the above is only rigorously true for affine space $\mathbb{C}^n$ since besides this case, K\"ahler cones are singular. At the singularity the space of forms is not defined, and only the $p=0$ summand of (\ref{scindexcharacterform}) is well defined. For a generic K\"ahler cone, it is necessary to regularise in order to define the index. In the following sections of this work, the primary aim will be to substantiate the supposition that the Dolbeault cohomology with respect to Zariski topology on the space obtained by an (equivariant) resolution of singularities is the appropriate regularisation. A \textit{resolution of singularities} here is a proper birational morphism $\pi : \tilde{X} \rightarrow X$ such that $\tilde{X}$ is non-singular. Henceforth we adopt this notation for the (un)resolved space and resolution. We also require $\pi$ to be equivariant with respect to the action of the holomorphic isometry generated by $D^I$ and the other global symmetries we grade by. We then define the regularised superconformal index as:
\begin{equation}\label{regularised index}
	\mathcal{Z}_{\tilde{X}}(\tau,\tilde{y},Z) = \sum^{d_{\mathbb{C}}}_{p,q = 0} (-)^{p+q-d_{\mathbb{C}}}\, \tilde{y}^{p-d_{\mathbb{C}}/2}\, \mbox{Tr}_{H^{q}(\tilde{X}; A^{p}(\tilde{X}))}\left(\tau^{D^{I}}\prod_{i}z_i^{\mathcal{J}_i}\right)
\end{equation}
Note that:
\begin{equation}
	\chi(A^{p}(\tilde{X})) \equiv \sum_{i=0}^{d_{\mathbb{C}}} (-)^i \mbox{ch}_{T}H^i\left(\tilde{X}, A^{p}(\tilde{X})\right)
\end{equation}
is the definition of the equivariant Euler character, where equivariant means with respect to the isometry algebra/group we grade by in the index. The superconformal index can then be expressed as the equivariant Hirzebruch $\chi_{\shortminus \tilde{y}}$ genus of $\tilde{X}$: 
\begin{equation}
\begin{split}
	\mathcal{Z}_{\tilde{X}}(\tau,\tilde{y},Z) &= \sum^{d_{\mathbb{C}}}_{p = 0} (-)^{p-d_{\mathbb{C}}}\, \tilde{y}^{p-d_{\mathbb{C}/2}} \chi(A^{p}(X))\\
	&\equiv (-1)^{d_{\mathbb{C}}}\tilde{y}^{-d_{\mathbb{C}}/2} \chi_{\shortminus\tilde{y}}\left(\tilde{X}\right)
\end{split}
\end{equation}
The questions that then need to be resolved are the following. Does an equivariant resolution of the K\"ahler cone $X$ exist, and further given two non-isomorphic equivariant resolutions do they have the same superconformal index? The latter point is equivalent to saying that there is an invariant $\mathcal{Z}(X)$ associated to the K\"ahler cone $X$ which is independent of resolution. Further we must show that the regularised index $\mathcal{Z}(\tilde{X})$ is consistent with the form of (\ref{scindexexpansion}) i.e. the representation theory of $\mathfrak{u}(1,1|2)$. In the remaining part of this work we give evidence for this hypothesis by showing consistency in different limits of the index, and for the particular case of toric Calabi-Yau 3-folds (see section \ref{sectiontoriccalabiyau3fold}), the full invariance of the superconformal index under the canonical crepant resolutions. 

\subsection{Consistency Checks}
In this section we provide some necessary conditions for the (regularised) superconformal index to satisfy, due to its definition pre-regularisation as a character of the $\mathfrak{u}(1,1|2)$ superalgebra. From (\ref{scindexexpansion}), we see that the following must be true:
\begin{itemize}
	\item Writing: 
	\begin{equation}
		\mathcal{Z}(\tau ,\tilde{y} ,Z) = \sum_{a,b} \alpha_{a,b}(Z) \tau^{a} \tilde{y}^{b}
	\end{equation}
	then it must be that the only monomials with non-zero coefficients are those with $a \geq 0 $ or rather:
	\begin{equation}\label{consistencycheck1}
		a < 0  \quad \Rightarrow \quad \alpha_{a,b} = 0
	\end{equation}
	
	\item Another consistency check arises from the positivity of the count of protected representations (\ref{protectedreps}). From (\ref{scindexexpansion}), we have that:
	
	\begin{equation} \label{consistencycheck2}
		\lim_{\tilde{y}\rightarrow \infty} \tilde{y}^{-\frac{d_{\mathbb{C}}}{2}} \mathcal{Z} = \sum_{d=0}^{\infty}\tau^d N\left(S\left(0,d,d_{\mathbb{C}}/2\right)\right)(Z) \quad\in\quad \mathbb{Z}^{+}\left[Z,\tau\right]
	\end{equation}
	i.e. the above limit must yield a polynomial in $Z$ and $\tau$ with positive integer coefficients. \\	
\end{itemize}
Note that in the hyperK\"ahler case \cite{doreysingleton} there is an isomorphism:
\begin{equation}
H^{p,q}(X,s) \cong H^{d_{\mathbb{C}}-p,q}(X,s)
\end{equation}
which follows from the fact that $H^{*,q}(X,s) $ forms a module for the $\mathfrak{su}(2)$ subalgebra generated by the raising operator corresponding to wedging with the holomorphic symplectic form defined on hyperK\"ahler manifolds: $\wedge \omega_{\mathbb{C}} = \wedge (\omega^1 + i\omega^2)$, and the symmetry of $\mathfrak{su}(2)$ representations. This gives a $y \rightarrow 1/y$ symmetry in the index, where $y=\tilde{y}\tau$. For further details see \cite{andrewthesis}. This, combined with the existence of more protected representations gives markedly more consistency checks in the hyperK\"ahler case.

\section{Regularisation and Computation}
\subsection{Regularisation and Localisation}
We now turn to the issue of resolving the K\"ahler cone and computing the regularised index. Since there are no theorems describing resolutions of a generic K\"ahler cone, we describe consistency in a variety of cases, in sections \ref{sectionregularcones}, \ref{sectioncalabiyaucones} and \ref{sectiontorickahlercones}. We first make the definition:
\begin{definition}
	A \textit{resolution} of a variety $X$ is a morphism of varieties $f : Y \mapsto X$ which is proper and birational, with $Y$ non-singular. We additionally impose the requirement that the resolution is an isomorphism over a non-singular locus $X_{reg}$.
\end{definition}
These conditions also imply that $f$ is proper, surjective and generically finite. In our case we also require that the resolution of singularities is equivariant with respect to the $\mathbb{C}^*$ action generated by $D^{I}$ and any of the mutually commuting global symmetries. Now we state some key theorems in showing the consistency check (\ref{consistencycheck2}):
\begin{theorem}\label{grauertriemenschneidervanishing}
	(Grauert-Riemenschneider \cite{grauertriemenschneider}): If $f : Y \mapsto X$ is any projective, birational morphism of varieties, with Y non-singular, then the higher derived direct images $R_{i}(f_{*}\omega_Y)) = 0 \quad \forall \,i \geq 1$ where $\omega_Y =  \left(\Omega^{1,0}(Y)\right)^{\text{dim}_{\mathbb{C}}(Y)}$ is the canonical sheaf of $Y$. 
\end{theorem}
\begin{theorem}\label{takegoshi}
	(Takegoshi Vanishing \cite{takegoshi}): If $f : Y \mapsto$ X is a proper surjective morphism of complex spaces and Y non-singular and bimeromorphic to a K\"ahler manifold, then $R_{i}(f_{*}\omega_Y)) = 0 \quad \forall\, i \geq  \mbox{dim}_{\mathbb{C}}(Y)-\mbox{dim}_{\mathbb{C}}(X)$.
\end{theorem}
Note that the last statement is in the analytic category. Note that when $X$ is an affine variety and either of the above theorems apply, then:
\begin{equation}
	H^{i}(Y, \omega_Y) = H^{i}(X, f_{*}\omega_Y) = 0 \qquad \forall i >0
\end{equation}
where the first equality follows from the above vanishing theorems, and the second from Serre vanishing since $X$ is an affine variety. From (\ref{regularised index}), it follows then that the highest coefficient of $\tilde{y}$ corresponding to consistency check (\ref{consistencycheck2}) reduces to:
\begin{equation}
	\mbox{Tr}_{H^{0}(\tilde{X};\, \omega_{\tilde{X}})}\left(\tau^{D^I}\prod_i z_{i}^{\mathcal{J}_i}\right)
\end{equation}
and is thus manifestly positive.\\

The equivariant Grothendieck-Riemann-Roch-Hirzebruch-Atiyah-Singer index theorem can then be used to compute the regularised index as an integration over equivariant characteristic classes, if we grade by an element of a compact Lie group. Of course, we have chosen to grade by $T =\mathbb{T}^s$. Further Atiyah-Bott localisation reduces the computation to an integration over the fixed point locus over the group action. The authors used \cite{pestunlocalisationingeometry} as reference. Identifying $g \in T=\mathbb{T}^s$ with the group element corresponding to $\tau^{D^{I}}\prod_{i}z_i^{\mathcal{J}_i}$, in the case of isolated $T$-fixed points, we obtain the equivariant Lefschetz formula:
\begin{equation}
	\chi(A^{p}(\tilde{X}))  = \sum_{x \in \tilde{X}^T} \frac{\mbox{ch}_{A^{p}_{x}}(g)}{\mbox{det}_{T^{*}_{x}}(1-g)}  = \sum_{x \in \tilde{X}^T} \mbox{ch}_T\left(\Lambda^{p}T^{*}_{x}, \tau, Z\right) \mbox{PE}\left[\mbox{ch}_T\left(T^{*}_{x}, \tau,Z\right)\right]
\end{equation}
Where $\tilde{X}^T$ is the set of $T$ fixed points of $\tilde{X}$ and $T^{*}(\tilde{X})$ is the cotangent bundle. Here PE denotes the plethystic exponential.\footnote{The plethystic exponential is defined as: $\text{PE}[f(t_1,\dots,t_n)]:=\text{exp}\left(\sum_{r=1}^\infty\frac{f(t_1^r,\dots,t_n^r)}r\right)$. A useful identity is: $\text{PE}\left[\sum_i t_i-\sum_j
	s_j\right]=\frac{\prod_j(1-s_j)}{\prod_i(1-t_i)}$, where $t_i$ and $s_i$ are monomials.} The superconformal index can then be computed as:
\begin{equation}\label{scindexlocalisationformula}
	\begin{split}
		\mathcal{Z}_{\tilde{X}}(\tau,\tilde{y},Z) &= \sum^{d_{\mathbb{C}}}_{p = 0} (-)^{p-d_{\mathbb{C}}}\, \tilde{y}^{p-d_{\mathbb{C}}/2} \chi_T(A^{p}(X)) \\
		&=  \sum^{d_{\mathbb{C}}}_{p = 0} (-)^{p-d_{\mathbb{C}}}\, \tilde{y}^{p-d_{\mathbb{C}}/2}\sum_{x \in \tilde{X}^T} \mbox{ch}_T\left(\Lambda^{p}T^{*}_{x}, \tau, Z\right) \mbox{PE}\left[\mbox{ch}_T\left(T^{*}_{x}, \tau, Z\right)\right] \\
		&= (-)^{d_{\mathbb{C}}}\tilde{y}^{-\frac{d_{\mathbb{C}}}{2}}\sum_{x \in \tilde{X}^T} \mbox{PE}\left[\left(1-\tilde{y}\right)\mbox{ch}_T\left(T^{*}_{x}, \tau, Z\right)\right]\\
		&=  (-)^{d_{\mathbb{C}}}\tilde{y}^{-\frac{d_{\mathbb{C}}}{2}}\sum_{x \in \tilde{X}^T} \prod_{\alpha \in \mathcal{W}_{x}} \frac{1-\tilde{y}\,m_{\alpha}(\tau,Z)}{1-m_{\alpha}(\tau,Z)}
	\end{split}
\end{equation}
Here $\mathcal{W}_x$ is the collection of $T$ weights of the isotropy action on $T_{x}^*(\tilde{X}) $ where a given weight is given by $\alpha = (\alpha_0, \gamma)$, where $\alpha_0$ is the weight under $D^{I}$ and $\gamma$ the weight of the remaining part of the torus $T$. If $T$ is an $s$-dimensional torus, then $\gamma$ is an $(s-1)$-vector. Here $m_{\alpha}(\tau,Z)= \tau^{\alpha_0}z_i^{\gamma_i}$. If the Reeb vector is appropriately normalised, then $\alpha_0 \in \mathbb{Z}$. The localisation formula gives the index as a rational function. Note that holomorphic functions (sections of the structure sheaf) on $X$ are non-negatively graded under the Reeb vector \cite{collinsthesis}. This dictates that the above rational function should be expanded in small positive $\tau$, thus verifying consistency check (\ref{consistencycheck1}).

\subsection{Singlets and Poincar\'e Polynomials}\label{sectionreeblimit}
Before we consider consistency in a variety of classes of K\"ahler cones, we first make some comments about the $\tau \rightarrow 0$ limit of the superconformal index. This limit is interesting as it yields the sector of the index containing information about the singlet states. From (\ref{scindexexpansion}) this is: 	
\begin{equation}\label{singletsector}
\lim_{\tau \rightarrow 0} \mathcal{Z}_{X}(\tau,\tilde{y} ,Z) = \sum_{r = -d_{\mathbb{C}}/2}^{ d_{\mathbb{C}}/2}\tilde{y}^{r}I^{0r} 
\end{equation}
We will see that this sector has a nice geometric interpretation in the target space. We work in the case where an equivariant resolution of the K\"ahler cone exists, and the fixed points of the torus action $T$ are isolated.\\

The only fixed point of the Reeb vector on the unresolved K\"ahler cone $X$ is the origin, therefore by equivariance the $D^{I}$-fixed submanifold on the resolved space $\tilde{X}$ (defined by the equivariant lift of the action on $X$) must be contained in $\pi^{-1}(\{o\})$. Then since $\pi$ is proper the fixed submanifold is compact. It could generically be a disjoint union of connected components, which we will label by $\rho$:
\begin{equation}
\tilde{X}^{D^{I}}= \bigsqcup_{\rho}\tilde{X}(\rho)
\end{equation}
The fixed point set of $T$ on $\tilde{X}$ is partitioned amongst the $\tilde{X}(\rho)$. Now consider the $\tau \rightarrow 0$ limit of the superconformal index:
\begin{equation}\label{reeblimitofindex}
\begin{split}
\lim_{\tau \rightarrow 0} \mathcal{Z}_{\tilde{X}}(\tau,\tilde{y},Z) 
&=  (-)^{d_{\mathbb{C}}}\tilde{y}^{-\frac{d_{\mathbb{C}}}{2}}\sum_{x \in \tilde{X}^T} \tilde{y}^{\,|
	\{\alpha: \alpha_{0}<0\}|} \prod_{\alpha \in \mathcal{W}_{x}| \alpha_{0}=0} \frac{1-\tilde{y}\,m_{\alpha}(Z)}{1-m_{\alpha}(Z)}\\
&=  (-)^{d_{\mathbb{C}}}\tilde{y}^{-\frac{d_{\mathbb{C}}}{2}}\sum_{\rho} \tilde{y}^{N_{\rho}^{(\shortminus)}} \sum_{x \in \tilde{X}^T \cap \tilde{X}(\rho)}  \prod_{\alpha \in \mathcal{W}_{x}| \alpha_{0}=0} \frac{1-\tilde{y}\,m_{\alpha}(Z)}{1-m_{\alpha}(Z)}\\
&= (-)^{d_{\mathbb{C}}} \sum_{\rho} \tilde{y}^{N_{\rho}^{(\shortminus)}\shortminus d_{\mathbb{C}}/2} \,\, \chi_{\shortminus \tilde{y}}(\tilde{X}(\rho))(Z)
\end{split}
\end{equation}
where $N_{\rho}^{(\shortminus)}$ is the number of $(-)$-attracting directions on a given connected component of the $D^{I}$-fixed submanifold of $\tilde{X}$. This gives the $\tau\rightarrow 0$ limit as the weighted sum of equivariant $\chi_{\shortminus\tilde{y}}$ genera of the connected components of the fixed point submanifold of the lifted action of the Reeb vector on the resolved space, graded by the $(s-1)$-torus with fugacities $Z=\{z_i\}$. \\

Here we have used the fact that the number of tangent directions negatively charged under the Reeb vector is the same at each of the $T$-fixed points lying in a given connected component $\tilde{X}(\rho)$. For quasi-regular and regular cones, this is because the tangent space at each point in a given connected component $\tilde{X}(\rho)$ forms a representation for the isotropy action of the $U(1)$ (complexified to $\mathbb{C}^*$) generated by the Reeb vector, so the weights of the action are quantised. Thus the weights are invariant on the smooth connected component, and in particular are the same at fixed points of the whole $T$-action. For irregular Reeb vectors the closure of the orbits is at least a two-torus. The fixed points of the Reeb vector coincide with those of the two-torus, whose eigenvalues are quantised on a given connected component. Thus the same holds for irregular Reeb vectors. \\

In \cite{pinglirigiditydolbeault} it was shown that in fact the equivariant $\chi_{-\tilde{y}}$ genus is in fact independent of $Z$, essentially using analyticity arguments. This means it is equal to the usual (ungraded) $\chi_{-\tilde{y}}$ genus. Further, in \cite{FeldmanHirzebruchgenus} it was shown that for a symplectic manifold $\mathcal{M}$ admitting a Hamiltonian circle action with isolated fixed points:
\begin{equation}
\chi_{\shortminus q}(\mathcal{M}) = P_{\sqrt{q}}(\mathcal{M})
\end{equation}
where the right hand side is the Poincar\'{e} polynomial defined as:
\begin{equation}
P_{q}(\mathcal{M}) = \sum_{i}^{\text{dim}_{\mathbb{R}}(\mathcal{M})} b_i q^i
\end{equation}	
and $b_i$ are the Betti numbers of $\mathcal{M}$. Further $\mathcal{M}$ has no odd-dimensional homology. These prerequisites are met by $\tilde{X}(\rho)$ since they are K\"ahler, and all symplectic circle actions on K\"ahler manifolds are Hamiltonian if they have non-empty fixed points. \\

Therefore the $\tau \rightarrow 0$ limit of the superconformal index, i.e. the index of the singlet sector, is given by the weighted sum of the Poincar\'{e} polynomials of the connected components of the fixed point submanifold of the equivariant lift of the Reeb vector on the resolved space:
\begin{equation}
\lim_{\tau \rightarrow 0} \mathcal{Z}_{\tilde{X}}(\tau, \tilde{y},Z) 
= (-)^{d_{\mathbb{C}}} \sum_{\rho} \tilde{y}^{N_{\rho}^{(\shortminus)}\shortminus d_{\mathbb{C}}/2} \,\, P_{\sqrt{\tilde{y}}}\left(\tilde{X}(\rho)\right) = \sum_{r\in \mathbb{Z}/2}\tilde{y}^{r}I^{0r} 
\end{equation}
Thus $ \lim\limits_{\tau \rightarrow 0}\left((-)^{d_{\mathbb{C}}}\tilde{y}^{d_{\mathbb{C}}/2}  \mathcal{Z}_{\tilde{X}(\tau, \tilde{y}, Z)}\right) \in \mathbb{Z}_{n \geq 0}\left[\tilde{y}\right]$, a finite polynomial. In fact, we have the further result:
\begin{proposition}\label{poincareproposition}
	Let $X$ be a K\"ahler cone such that the resolution $\tilde{X}$ has isolated fixed points under an $s\shortminus$torus $T=\mathbb{T}^s$ of holomorphic isometries such that $\mathcal{L}(\mathbb{T}^s) \equiv \mathfrak{t}_s \ni D^I$. Then the $\tau \rightarrow 0$ limit of the superconformal index on the cone is the same for all choices of Reeb vector in the Reeb cone, where $\tau$ is the fugacity grading by the choice of Reeb vector. Further, it is equal up to an overall factor  of $\tilde{y}^{-d_{\mathbb{C}}/2}$ to the Poincar\'e polynomial of $\pi^{-1}(\{o\})$.
\end{proposition}
\begin{proof}
	The requirement that non-constant holomorphic functions are graded positively under the Reeb vector $D^I$ dictates the expansion of the index given as a rational function. Since we have shown the $\tau \rightarrow 0$ limit of the index is $Z$ independent, the index can be expanded:
	\begin{equation}
	A_0(\tilde{y}) + \tau A_1(\tilde{y},Z) + \tau^2 A_2(\tilde{y},Z)  +...
	\end{equation}
	Where $A_0(\tilde{y})$ is a finite Laurent polynomial by the previous result. Under a rescaling of fugacities to obtain the superconformal index of the K\"ahler cone defined with respect to a different Reeb vector in the Reeb cone, $A_0(\tilde{y})$ is invariant, and any potential rescalings of monomials in $A_{i>0}(\tilde{y},Z)$ to obtain a $\tau$-independent monomial must depend on the fugacities $Z$. However we have already shown such terms cannot contribute to the $\tau\rightarrow0$ limit of any superconformal index. Thus the  $\tau\rightarrow0$ limit of the superconformal index is invariant under choice of Reeb vector. \\
	
	For the second point we closely follow the argument of \cite{doreybarns-graham}. We choose a generic $\mathbb{C}^*$ action $\lambda : \mathbb{C}^* \xhookrightarrow{} \mathbb{T}^s$ such that the fixed points of the generic action coincide with the fixed points of $\mathbb{T}^s$, and such that the action is generated by a vector lying in the Reeb cone so that $\lim_{t \rightarrow 0}\lambda(t)  = 0$. For instance, one could choose a $\mathbb{C}^*$ by making the following relabelling of fugacities (see section \ref{sectiontorickahlercones}): 
	\begin{equation}\label{poincarepolyrescaling}
	\tau \mapsto t^{m}, \qquad z_i \mapsto t^{n_i}\tilde{z}_i \qquad m \gg n_1 > .... > n_{s-1} > 0
	\end{equation}
	which corresponds to an action lying in the Reeb cone. This is because it specifies up to constant rescaling a vector which is just a small deformation from the original choice of Reeb vector. We define the $(-)$-attracting set at each fixed point $x$ to be:
	\begin{equation}
	U_{x} \equiv \{ p \in \tilde{X} | \lim_{t\rightarrow \infty} \lambda(t) \cdot p =  x  \}
	\end{equation}
	This is also equal to the dimension of the tangent space which is negatively charged under $\lambda$. By properness and equivariance $\bigcup U_x = \pi^{-1}(\{o\})$. Theorems (3) and (4) of \cite{nakajimalectures} apply also for $\tilde{X}$ described above, implying the vanishing of odd homology of $\pi^{-1}(\{o\})$ and the fact that the even homology is freely generated. Further, each fixed point of the generic action generates a single generator with homology degree given by the dimension of the $(-)$-attracting set:\footnote{Actually, one could equally use the (+)-attracting set due to Poincar\'e duality $b_{k} = b_{2d_{\mathbb{C}}-k}$}
	\begin{equation}
	\sum_{i=0}^{\text{dim}_{\mathbb{C}}\left( \pi^{-1}(\{o\})\right)  } \dim H_{2i}(\pi^{-1}(\{o\}))q^i = \sum_{x \in \tilde{X}^{T}}q^{\text{dim}U_{x}}
	\end{equation}
	Making the rescaling of fugacities (\ref{poincarepolyrescaling}) in the index and taking the $t \rightarrow 0$ limit one has, similarly to (\ref{reeblimitofindex}):
	\begin{equation}\label{poincarepolynomialcore}
	\begin{split}
	\lim_{t \rightarrow 0} \mathcal{Z}_{\tilde{X}}(t,\tilde{y}, \tilde{Z}) 
	&=  (-)^{d_{\mathbb{C}}}\tilde{y}^{-\frac{d_{\mathbb{C}}}{2}}\sum_{x \in \tilde{X}^T} \tilde{y}^{\text{dim}U_x} = (-)^{d_{\mathbb{C}}}\tilde{y}^{-\frac{d_{\mathbb{C}}}{2}} P_{\sqrt{\tilde{y}}}(\pi^{-1}(\{o\}))
	\end{split}
	\end{equation}
	There are no plethystic contributions since the action is generic. 
\end{proof}
\begin{corollary}
	On a given K\"ahler cone, all possible choices of metric lead to the same result for the sector of the index containing information about the singlet states (\ref{singletsector}). 
\end{corollary}
The equation (\ref{poincarepolynomialcore}) implies something interesting about the singlet representations $S'_1(d_{\mathbb{C}}/2)$ which are protected. Note that the highest power of $\tilde{y}$ on the right hand side of (\ref{poincarepolynomialcore}) is $\text{dim}_{\mathbb{C}}\left( \pi^{-1}(\{o\})\right)-d_{\mathbb{C}}/2$ which is strictly less than $d_{\mathbb{C}}/2$. Comparing with (\ref{singletsector}) this rules out the existence of the singlet states which are protected for the cones considered in proposition \ref{poincareproposition} (i.e. when the equivariant resolution $\tilde{X}$ has isolated fixed points under a torus action). Of course there could still be singlet states, but they are not protected (as they have $r<d_{\mathbb{C}}/2$) and so we cannot determine their degeneracy uniquely using the index.\\

We will later see this played out explicitly in the case of $A$-type quiver varieties. Now we return to consistency of the index for various classes of K\"ahler cones.

\subsection{K\"ahler Cones Over Regular Sasaki Manifolds}\label{sectionregularcones}
For K\"ahler cones over regular smooth Sasaki manifolds, Sparks, Martelli and Yau showed \cite{sparksmartelliyauvolume} that there always exists a natural equivariant smooth resolution of $X$. In this case the Sasakian link $Y$ can be exhibited as the total space of a principal circle line bundle $\pi: Y \rightarrow V$ over a K\"ahler manifold $V$, which turns out to be a normal projective algebraic variety. A resolution $W$ is then the total space of the associated complex line bundle $\mathcal{L} \rightarrow V$ of the principal circle bundle $Y \rightarrow V$. $W \backslash V$ is biholomorphic to the complement of the tip of the cone, thus $W$ is birational to $X$. If we assume as in \cite{sparksmartelliyauvolume} that on $W$ there exists a 1-parameter family of $\mathbb{T}^s$- invariant K\"ahler metrics on W which approach the K\"ahler cone metric as the parameter is taken to 0, then theorem \ref{takegoshi} can be applied, and one can therefore show that indeed the highest coefficient of $\tilde{y}$ yields a positive coefficient Laurent polynomial in the fugacities as required to be consistent with the protection of representations. 

\subsection{Ricci-Flat Calabi-Yau Cones}\label{sectioncalabiyaucones}
In this section we consider the case when $X$ is Ricci-flat, so by definition the link $Y$ is \textit{Sasaki-Einstein}. We also assume that the singularity $\{o\}$ is Gorenstein, admitting a nowhere vanishing holomorphic $(d_{\mathbb{C}},0)$ form. This is the same as requiring that the canonical bundle on $X$ is holomorphically trivial and therefore $X$ is Calabi-Yau. In this case it is possible to say more about the consistency of the definition of the superconformal index. Supposing that the Sasaki base $Y$ is smooth so that the singularity at the tip of the cone is isolated, then $X$ is a rational singularity and is normal \cite{Coeveringcrepantresolutions}. In particular this means that if $\pi: \tilde{X} \mapsto X$ is a resolution then $R^i \pi_{*} \mathcal{O}_{\tilde{X}} = 0 \quad \forall i > 0 $ and $\pi_{*} \mathcal{O}_{\tilde{X}} \simeq \mathcal{O}_{X}$, where $\mathcal{O}$ denotes the structure sheaf, and this holds for all resolutions. If the resolution $\pi$ is in addition equivariant, then one has that:
\begin{equation}
	H^0\left(\tilde{X}, \mathcal{O}_{\tilde{X}}\right) \simeq H^0\left(X, \mathcal{O}_{X}\right) \qquad H^i\left(\tilde{X}, \mathcal{O}_{\tilde{X}}\right)  = H^i\left(X, \mathcal{O}_{X}\right) = 0 \quad \forall i > 0 
\end{equation}
Where the first ismorphism is not just as abelian groups but as $T$-modules. Note that then from (\ref{regularised index}):
\begin{equation} \label{hilbertseries}
	\begin{split}
		\lim\limits_{\tilde{y} \rightarrow 0 } \tilde{y}^{d_{\mathbb{C}}/2} \mathcal{Z}_{\tilde{X}} &= \sum_{q=0}^{d_{\mathbb{C}}} (-)^{q-d_{\mathbb{C}}} \mbox{Tr}_{H^q(\tilde{X}; \mathcal{O}_{\tilde{X}})}\left(\tau^{D^I}\prod_{i}z_i^{\mathcal{J}_i}\right)\\
		& = (-)^{d_{\mathbb{C}}} \mbox{Tr}_{H^{0,0}(X)} \left(\tau^{D^I}\prod_{i}z_i^{\mathcal{J}_i}\right)\\
		& =  (-)^{d_{\mathbb{C}}} \mbox{HS}(X)
	\end{split}
\end{equation}
where $\mbox{HS}(\tilde{X})$ is the Hilbert series of the unresolved space $X$, and hence this limit of the superconformal index is invariant under choice of equivariant resolution. \\

If in addition we assume the resolution is \textit{crepant}, i.e. that the resolution has trivial canonical bundle $\pi^* \omega_{X} \simeq \omega_{\tilde{X}} \simeq \mathcal{O}_{\tilde{X}}$, then $\tilde{X}$ admits a Ricci-flat K\"ahler metric (in fact it does so in every K\"ahler class), as was shown simultaneously in \cite{Coeveringcrepantresolutions3, goto}. These are the natural resolutions to consider as they preserve the Calabi-Yau property. Then Takegoshi vanishing (theorem \ref{takegoshi}) applies and the required representations are protected as required. In fact, one can see directly that the higher direct images vanish, as $\mathcal{O}_{\tilde{X}} \cong \omega_{\tilde{X}}$ due to the triviality of the canonical bundle, and then applying the result for rational singularities above. In addition since $X$ is a rational singularity \cite{CoeveringCrepantResolutions2} $\pi_{*}\omega_{\tilde{X}} = \omega_{X}$, so replacing $\mathcal{O}$ with $\omega$ everywhere above we see that if the resolution is equivariant the character of the group action on the canonical sheaf is invariant under resolution.\footnote{Whilst $\omega_{\tilde{X}}$ can be identified as the canonical sheaf on $\tilde{X}$, $\omega_{X}$ on the singular space strictly speaking refers to the dualizing sheaf of $X$, which is $\omega_{X} \cong i_{*}(\omega_{X\textbackslash\{o\}})$ where $i : X\textbackslash \{o\} \rightarrow X$ is the inclusion. See \cite{Coeveringcrepantresolutions} for details.} Thus we have shown that the top power of $\tilde{y}$ in the superconformal index corresponding to the count of protected representations (\ref{consistencycheck2}) is invariant under choice of equivariant crepant resolution for a Calabi-Yau cone. In fact we can say even more:
\begin{proposition}
	The isomorphism of the structure sheaf and canonical sheaf of the equivariant crepant resolution of a Ricci-flat Calabi-Yau cone is a $T$-graded isomorphism.
\end{proposition}
\begin{proof}
	Let $n={d_{\mathbb{C}}}$. The Gorenstein property implies the existence of a nowhere-vanishing holomorphic $(n,0)$-form $\Omega$ on $X\textbackslash \{o\}$, which must satisfy 
	\begin{equation}\label{calabiyauproperty}
	\frac{i^n}{2^n} (-)^{n(n-1)/2} \Omega \wedge \bar{\Omega}= e^{f} \frac{1}{n!} \omega^n
	\end{equation}
	where $\omega$ is the K\"ahler potential on $X$ and $f$ is a real function, by bidegree arguments. Since $\omega^n$ gives the determinant of the K\"ahler potential/metric, one can show easily that Ricci-flatness implies $f=0$. Using the fact that $\omega$ is homogeneous degree 2 under $r\frac{\partial}{\partial r}$, (\ref{calabiyauproperty}) implies that:
	\begin{equation}
	\mathcal{L}_{r\frac{\partial}{\partial r}} \Omega = d_{\mathbb{C}} \Omega \qquad  \Rightarrow \qquad \hat{D}^I \Omega = - i \mathcal{L}_{I \left(r\frac{\partial}{\partial r}\right)} \Omega =   d_{\mathbb{C}} \Omega
	\end{equation}
	It is also easy to see that $\Omega$ must have definite eigenvalue under any other holomorphic isometry. Let $V$ be the real holomorphic vector field  generating the holomorphic isometry, that is $\bar{V} = V$ and $\mathcal{L}_{V} I = 0$. This implies $\mathcal{L}_{V}$ preserves the bidegree of forms, and commutes with $\bar{\partial}, \,\, \partial$. Thus $\mathcal{L}_V \Omega = \alpha(z) \Omega$ for some holomorphic function $\alpha(z)$, and $\mathcal{L}_V \bar{\Omega} = \bar{\alpha}(\bar{z}) \bar{\Omega}$. Acting with $\mathcal{L}_{V}$ on both sides of (\ref{calabiyauproperty}) with $f=0$ implies that:
	\begin{equation}
	\alpha(z) = \bar{\alpha}(\bar{z}) =\alpha\quad  \text{(constant)}
	\end{equation}
	We also denote by $\Omega$ the extension of $\pi^* \Omega$ to a nowhere vanishing $n$-form on $\tilde{X}$ (see for example \cite{Coeveringcrepantresolutions}). By equivariance of $\pi$ it has the same eigenvalues under the holomorphic isometry group. $\Omega$ then defines an isomorphism $\wedge \Omega : \mathcal{O}_{\tilde{X}} \xrightarrow{\sim}  \omega_{\tilde{X}}$, e.g. if $f \in \Gamma(\tilde{X}, \mathcal{O}_{\tilde{X}})$ is homogeneous degree $d$, then $f\Omega \in \Gamma(\tilde{X}, \omega_{\tilde{X}})$ is homogeneous degree $d+d_{\mathbb{C}}$, and similarly for the other holomorphic isometries. 
\end{proof}

This isomorphism implies that for Ricci-flat Calabi-Yau cones admitting equivariant crepant resolutions, there is a factorisation between the top power $\tilde{y}^{d_{\mathbb{C}}/2}$ in the index and the bottom power $\tilde{y}^{-d_{\mathbb{C}}/2}$:
\begin{equation}\label{factorisationproperty}
	(-)^{d_{\mathbb{C}}}\left(\frac{\lim\limits_{\tilde{y}\rightarrow \infty} \tilde{y}^{-d_{\mathbb{C}}/2} \mathcal{Z}(\tau, \tilde{y}, Z) }{\lim\limits_{\tilde{y} \rightarrow 0 } \tilde{y}^{d_{\mathbb{C}}/2} \mathcal{Z}(\tau, \tilde{y},Z)} \right)= \tau^{d_{\mathbb{C}}}(...)
\end{equation}
where $(...)$ corresponds to a monomial in $Z = \{z^i\}$ fugacities corresponding to the global holomorphic isometries of $X$ and the whole right hand side corresponds to the character of $\Omega$. Note that therefore on Ricci-flat Calabi-Yau cones this precludes the existence of the special $S'_1(d_{\mathbb{C}})$ states which are protected, since these correspond to top holomorphic degree forms of degree $0$ under $D^I$. \\

It is also interesting to consider the space of K\"ahler cone metrics varied over in the volume minimisation procedure in \cite{sparksmartelliyauvolume}, the AdS/CFT dual counterpart of a–maximisation in 4d superconformal field theories. These are also rational \cite{Coeveringcrepantresolutions} and have trivial canonical sheaves. Therefore the invariance and protection of the highest power of $\tilde{y}$, and the invariance of the lowest power of $\tilde{y}$ under equivariant crepant resolutions hold for these families of K\"ahler metrics also.\\

\subsection{Toric K\"ahler Cones}\label{sectiontorickahlercones}
In this section we consider toric K\"ahler cones $(X,g)$, so that $X$ is an affine toric variety of complex dimension $n$ equipped with the holomorphic and isometric action of a torus $\mathbb{T}^{n}$. Since the manifold is complex this can be extended to a holomorphic $\left(\mathbb{C}^*\right)^n$ action. We also assume that the Sasaki base $Y$ is smooth so that the tip of the cone $\{o\}$ is an isolated singularity. Toric geometry is an extremely well-developed subject and we refer the reader to the comprehensive review \cite{Coxtoricvarietiesandtoricresolutions}. This paper will assume definitions, results and notations therein. \\

As an affine toric variety, the K\"ahler cone $X$ is specified by a strongly convex rational polyhedral cone $\sigma \subset \mathbb{R}^n$ specified by a set of vectors $\{v_{\gamma}\}\in \mathbb{Z}^n \subset \mathbb{R}^n$ comprising its edges - the toric data.\footnote{Note that our convention is such that the cone and the dual cone are switched when compared to \cite{sparksmartelliyauvolume}.} These vectors can be taken to be \textit{primitive} - that is $v_{\gamma}$ cannot be written as $kv'_{\gamma}$ for some integer $k>1$ and $v'_{\gamma} \in \mathbb{Z}^n$. Unless the cone is flat, the toric variety will have singularities. \\

An affine toric variety may be resolved as follows. A \textit{fan} is a collection $\Sigma$ of strongly convex rational polyhedral cones $\{\tau\}$  such that if $\tau \in \Sigma$ any face of $\tau$ is also in $\Sigma$, and the intersection of any two cones is a face of each. We have from theorem 5.1 of \cite{Coxtoricvarietiesandtoricresolutions} that there exists a fan $\Sigma$ consisting of a collection of strongly convex rational polyhedral cones \textit{refining} $\sigma$, meaning that $\sigma$ and $\Sigma$ have the same support in $\mathbb{R}^n$ and that $\sigma$ is one of the cones of $\Sigma$, so that the toric variety $X_{\Sigma}$ corresponding to the fan is a resolution of $X$. One can further prove the existence of a toric resolution $\phi: X_{\Sigma'} \rightarrow X$ such that $\phi$ is a projective resolution. This suggests $\Sigma'$ as the appropriate refinement, since then we can apply Grauert-Riemenschneider vanishing (theorem \ref{grauertriemenschneidervanishing}) to show that the appropriate representations are protected. Further, since the singularity is rational (and normal), much of the discussion in the previous section applies (despite $X$ not necessarily being Calabi-Yau). So the Hilbert series sector of the index is also is also invariant under resolution, indicating consistency. In section \ref{sectiontoriccalabiyau3fold} we consider the example of toric Calabi-Yau 3-folds for which an even higher level of consistency can be shown.\\

The toric data conveniently encodes the fixed point data, and the regularised index (\ref{scindexlocalisationformula}) can be easily computed. Note that the toric variety associated to the fan $\Sigma'$ is smooth if and only if each cone $\tau \in \Sigma$ is  generated by a subset of a $\mathbb{Z}$-basis of $\mathbb{Z}^n$. We then have that the varieties associated with each smooth $n$-dimensional cone satisfy $X_{\tau} \cong \mathbb{C}^n $, are Zariski open in $X_{\Sigma'}$ and are then glued together in an equivariant way to form the toric variety $X_{\Sigma'}$. The faces of $\tau$ have inward facing normals $\{u_{\alpha}\}$, $\alpha=1,...,n$  which can be taken to be primitive in the dual lattice to $\mathbb{Z}^n$. One can show the fixed points of the torus action correspond to the $n$-dimensional cones in the fan. The action of $t \in \left(\mathbb{C}^*\right)^n$ action acts in each neighbourhood (let $(z_1,...,z_n)$ be coordinates in the neighbourhood) as:
\begin{equation}
	(t_1,...t_n) \cdot (z_1,...,z_n) = (t^{u_1}z_1,...,t^{u_n}z_n)
\end{equation}
where $t^{u_{\alpha}} = t_1^{u_{\alpha}^{1}}...t_n^{u_{\alpha}^{n}}$. We will assume the cone is of Reeb type, i.e. that the Reeb vector lies in the Lie algebra $\mathfrak{t}_n$ of $\mathbb{T}^n$, and all cones we consider will be of this type. Grading by the whole $\mathbb{T}^n$, the superconformal index can be calculated using (\ref{scindexlocalisationformula}) as:
\begin{equation}\label{scindextoric}
	\mathcal{Z} = (-)^n\tilde{y}^{-n/2} \sum_{\tau  \in \Sigma'(n)}\prod_{\alpha=1,...n}\frac{1-\tilde{y} x^{u_{\alpha}}}{1-x^{u_{\alpha}}}
\end{equation}
where $\Sigma'(n)$ is the set of $n$-dimensional cones in $\Sigma'$, $u_{\alpha}$ the primitive normal vectors of the faces specifying $\tau$ and $x_j$, $j=1,...,n$ are the fugacities corresponding to the action of $\mathbb{T}^n$. To write the index in terms of grading by the Reeb vector, i.e. in terms of the fugacity $\tau$, a relabelling of fugacities must be performed as follows. Giving angular coordinates $\phi_i$ to the orbits of the $\mathbb{T}^n$ action, the Reeb vector may be expressed as:
\begin{equation}
	D^{I} = b_i \frac{\partial}{\partial \phi_i}
\end{equation}
with $b \in \mathbb{R}^n$. Note that if the Sasakian link is quasi-regular (or regular) then $b \in \mathbb{Q}^n$. Assuming $b_1 \neq 0 $ since at least one of the $b_i \neq 0$, then:
\begin{equation}
	\begin{split}
		x_1^{\frac{\partial}{\partial \phi_1}}...\,\,x_n^{\frac{\partial}{\partial \phi_n}} &= e^{\frac{\log{x_1}}{b_1}\left(b_1 \frac{\partial}{\partial \phi_1} + ... + b_n  \frac{\partial}{\partial \phi_n}\right)} \prod_{i=2,...,n} e^{\left(\log x_i - \frac{b_i}{b_1} \log x_1\right)\frac{\partial}{\partial \phi_i}} \\
		&= \left((x_1)^{\frac{1}{b_1}}\right)^{D^I} \prod_{i=2,...,n}\left(x_i x_1^{-\frac{b_i}{b_1}}\right)^{\frac{\partial}{\partial \phi_i}}
	\end{split}
\end{equation}
So in (\ref{scindextoric}) relabelling:
\begin{equation} \label{fugacityrelabelling}
	x_1 \mapsto \tau^{b_1} \qquad x_i \mapsto \tau^{b_i} z_i \quad \forall\, i=2,...n
\end{equation}
the index with grading by $\tilde{y}, \tau, \{z_i\}$ is obtained as in (\ref{regularised index}).

\section{Examples}
Here the superconformal index is computed in a variety of examples, and consistency with superconformal invariance is shown. 

\subsection{Toric Calabi-Yau 3-folds }\label{sectiontoriccalabiyau3fold}
We consider the subset of the cones considered in section \ref{sectiontorickahlercones} which are also Calabi-Yau 3-folds. This class of K\"ahler cones we consider are those for which the highest level of consistency of the regularisation can be shown. These are toric cones which are Calabi-Yau 3-folds and therefore have trivial canonical bundle. First we make some general statements about toric cones which are Calabi-Yau $n$-folds. The Calabi-Yau requirement is that the primitive edges $\{v_{\gamma}\}$ lie in a hyperplane of $\mathbb{R}^n$. One can then choose a basis of $\mathbb{Z}^n$ (corresponding to an $SL(n,\mathbb{Z})$ transformation) such that the primitive edge vectors can be placed in the form $v_{\gamma} = (1,w_{\gamma})$ where $w_{\gamma}\in \mathbb{Z}^{n-1}$. The convex hull of $\{w_{\gamma}\}$ defines a convex lattice polytope $P$. We call the projection of the cone onto this hyperplane the $\textit{toric diagram}$. The existence of Ricci-flat K\"ahler metrics on such affine toric varieties was proved in \cite{futaki}. Upon choosing this metric, this class of K\"ahler cones falls into the intersection of those considered in sections \ref{sectioncalabiyaucones} and \ref{sectiontorickahlercones}.\\

The natural resolutions to look at are those preserving the Calabi-Yau property, i.e. toric crepant resolutions. Following section \ref{sectiontorickahlercones} this is equivalent to a refinement $\Sigma$ of the cone $\sigma$ defining the toric variety $X$ such that primitive edge vectors of all cones in $\Sigma$ lie in the same hyperplane. This is equivalent to a basic lattice triangulation of $P$, meaning that the vertices of each simplex lie in $\mathbb{Z}^{n-1}$ and the vertices of each $(n-1)$-dimensional simplex in the triangulation generate a basis of $\mathbb{Z}^{n-1}$. For details see \cite{Coeveringcrepantresolutions}. For $n=3$ basic lattice triangulations of the 2-dimensional polytope $P$ coincide with maximal triangulations, which are triangulations such that the lattice points contained in each simplex are only its vertices, and these always exist. Thus for K\"ahler cones which are toric Calabi-Yau 3-folds, one can always find an equivariant toric crepant resolution. \\

The statements of sections \ref{sectioncalabiyaucones} and \ref{sectiontorickahlercones} apply then to this case. Actually, in the case of toric Calabi-Yau 3-folds we can say something stronger. The equivariant toric crepant resolution described is in general not unique, there exist multiple basic lattice triangulations of $P$. However, we claim that:

\begin{proposition}
	The regularised superconformal index (\ref{regularised index}) of toric Calabi-Yau cones is invariant under choice of toric crepant resolution, of which at least one exists.
\end{proposition}
\begin{proof}
	This is equivalent to proving that the equivariant $\chi_{-\tilde{y}}$ genus of the resolved spaces are the same. In 3 complex dimensions such crepant toric resolutions can be reached from another by a sequence of \textit{flop} transitions corresponding to re-triangulating convex quadrilaterals formed by neighbouring triangles \cite{kollarflops, roanflops}. Two examples can be seen in figure \ref{flopdiagram}. Since the flop changes the formula (\ref{scindextoric}) in a local way, one need only show the contribution of the convex quadrilateral is invariant under the flop. Pick's formula says that if a lattice polygon of area contains $b$ lattice points on its boundary and $g$ in its interior, then its area is given by $A= b/2 +g-1$. A convex quadrilateral formed from two triangles in the basic lattice triangulation has $b=4$ and $g=0$ therefore area $A=1$. In \cite{rabinowitzconvexlatticepolygons} it is proven that such convex lattice quadrilaterals with no interior points are always related via at most an $SL(2; \mathbb{Z})$ transformation and an integer translation in the lattice to the toric diagram of the conifold, that is the diagram with $\{w_{\alpha}\}=\{(0,0), (1,0), (1,1), (0,1)\}$. Such a transformation takes the form:
	\begin{equation}
	\begin{pmatrix}
	x\\ 
	y
	\end{pmatrix}
	\mapsto
	\begin{pmatrix}
	a & b \\
	c & d
	\end{pmatrix}
	\begin{pmatrix}
	x \\
	y
	\end{pmatrix}+
	\begin{pmatrix}
	e\\
	f\\
	\end{pmatrix}
	\end{equation}
	where $a,...,f \in \mathbb{Z}$ and $ad-bc=1$. The vertices of a generic convex lattice quadrilateral formed from two triangles in the triangulation therefore has vertices:
	\begin{equation}
	\{v_{\gamma}\}=\{(1,w_\gamma)\}=\left\{(1,e,f), (1,a+e,c+f),(1,a+b+e,c+d+f),(1,b+e,d+f)\right\}
	\end{equation}
	We can triangulate this by either joining $v_1$ and $v_3$ or $v_2$ and $v_4$. In 3 dimensions an easy way of determining the primitive normal vectors is by taking the cross product of the edge vectors defining the hyperplane. The cross product is guaranteed to be primitive from the condition $ad-bc=(a+b)d-b(c+d)=1$ ensuring that the e.g. $a$ and $b$ are coprime, as are $a+b$ and $c+d$. A convenient picture is the so-called \textit{pq-web}, which is the graph theory dual of the toric diagram obtained by mapping faces $\leftrightarrow$ vertices and edges to orthogonal edges. The external legs of the pq-web are then the projection of the primitive normal vectors onto the hyperplane with normal $(1,0,0)$, and the vertices of the pq-web correspond to the fixed points. For one of two possible triangulations this is illustrated in figure \ref{pqwebdiagram1}.
		\begin{figure}
		\begin{minipage}{0.5\textwidth}
			\centering
			\includegraphics[height=25mm,angle=0,  clip=true]{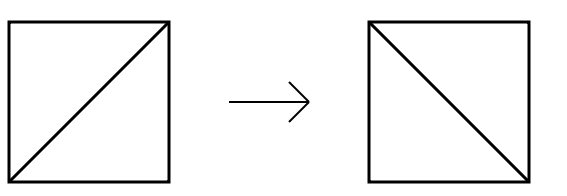} 
		\end{minipage}\hfill
		\begin{minipage}{0.5\textwidth}
			\centering
			\includegraphics[height=40mm,angle=0,  clip=true]{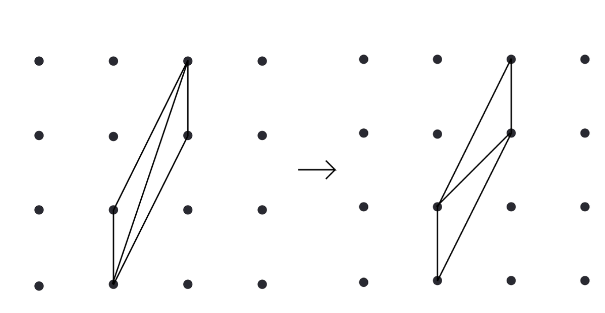} 
		\end{minipage}
		\caption{The conifold flop transition, and an example of a flop transition corresponding to a different convex quadrilateral}\label{flopdiagram}
	\end{figure}
	\begin{figure}
		\hspace*{-15mm}
		\begin{minipage}{0.5\textwidth}
			\centering
			\includegraphics[height=60mm,angle=0,  clip=true]{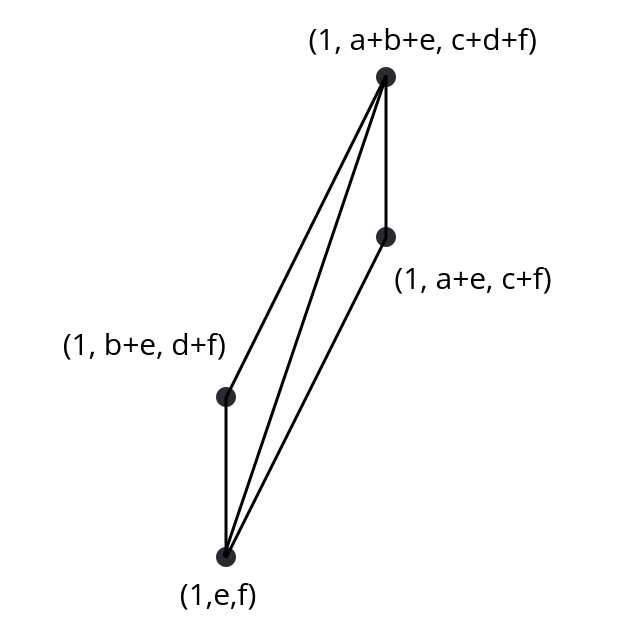} 
		\end{minipage}\hspace{-10mm}
		\begin{minipage}{0.5\textwidth}
			\centering
			\includegraphics[height=75mm,angle=0,  clip=true]{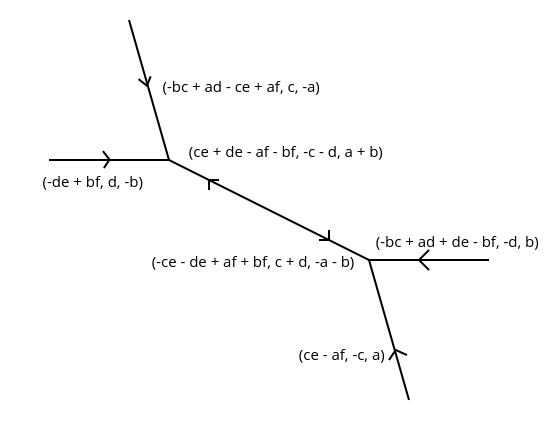} 
		\end{minipage}
		\caption{A triangulation of the generic convex quadrilateral with no interior lattice points, and its pq-web}\label{pqwebdiagram1}
	\end{figure}
	Given the pq-web the contribution to (\ref{scindextoric}) is easily calculated, and equally so is that of the alternate triangulation which has a different pq-web. Their contributions match and therefore the superconformal index is independent of resolution.
	\end{proof}
\subsubsection{The Conifold}
We consider now a simple example of a toric Calabi-Yau K\"ahler cone, the conifold, found commonly in the physics literature. Its diagram was specified above in figure \ref{flopdiagram}, and one of two triangulations corresponds to figure \ref{pqwebdiagram1} with $a=d=1$, $b=c=e=f=0$. Of course the index is invariant under choice of triangulation as proved above. The index is given by (\ref{scindextoric}):
\begin{equation}
	(-)^3 \tilde{y}^{-\frac{3}{2}}\left[\frac{\left(1-\frac{x_3 \tilde{y}}{x_2}\right) \left(1-x_2 \tilde{y}\right) \left(1-\frac{x_1 \tilde{y}}{x_3}\right)}{\left(1-\frac{x_3}{x_2}\right) \left(1-x_2\right) \left(1-\frac{x_1}{x_3}\right) }+\frac{\left(1-\frac{x_1 \tilde{y}}{x_2}\right) \left(1-\frac{x_2 \tilde{y}}{x_3}\right) \left(1-x_3 \tilde{y}\right)}{\left(1-\frac{x_1}{x_2}\right) \left(1-\frac{x_2}{x_3}\right) \left(1-x_3\right)}\right]
\end{equation}
In \cite{smyauamaximisation} the Reeb vector corresponding to the Ricci-flat metric was found via an extremisation procedure to be $b = (3,3/2,3/2)$ and so relabelling fugacities:
\begin{equation}
	x_1 \mapsto \tau^3 \qquad x_2 \mapsto \tau^{\frac{3}{2}}z_2 \qquad x_3 \mapsto \tau^{\frac{3}{2}}z_3
\end{equation}
the superconformal index $\mathcal{Z}\left(\tilde{y}, \tau, z_2, z_3 \right)$ is therefore:
\begin{equation}
	\begin{split}
		\mathcal{Z} &=(-)^3 \tilde{y}^{-\frac{3}{2}} \left[\frac{\left(1-\frac{\tilde{y} z_3}{z_2}\right) \left(1-\tau ^{3/2} \tilde{y} z_2\right) \left(1-\frac{\tau ^{3/2} \tilde{y}}{z_3}\right)}{\left(1-\frac{z_3}{z_2}\right) \left(1-\tau ^{3/2} z_2\right) \left(1-\frac{\tau ^{3/2}}{z_3}\right)}+\frac{\left(1-\frac{\tilde{y} z_2}{z_3}\right) \left(1-\frac{\tau ^{3/2} \tilde{y}}{z_2}\right) \left(1-\tau ^{3/2} \tilde{y} z_3\right)}{\left(1-\frac{z_2}{z_3}\right) \left(1-\frac{\tau ^{3/2}}{z_2}\right) \left(1-\tau ^{3/2} z_3\right)}\right]\\
		&= (-)^3 \tilde{y}^{-\frac{3}{2}}  \left(
		\begin{array}{c}
			\tau ^6 \tilde{y}^3 z_2 z_3-\tau ^3 \tilde{y}^3 z_2 z_3+\tau ^3 \tilde{y}^2-2 \tau ^{9/2} \tilde{y}^2 z_2 z_3^2-2 \tau ^{9/2} \tilde{y}^2 z_2-2 \tau ^{9/2} \tilde{y}^2 z_3 \\
			-2 \tau ^{9/2} \tilde{y}^2 z_2^2 z_3+\tau ^6 \tilde{y}^2 z_2 z_3+\tau ^3 \tilde{y}^2 z_2^2+\tau ^3 \tilde{y}^2 z_3^2+\tau ^3 \tilde{y}^2 z_2^2 z_3^2+3 \tau ^3 \tilde{y}^2 z_2 z_3\\
			+\tau ^3 \tilde{y}-2 \tau ^{3/2} \tilde{y} z_2 z_3^2-2 \tau ^{3/2} \tilde{y} z_2-2 \tau ^{3/2} \tilde{y} z_2^2 z_3-2 \tau ^{3/2} \tilde{y} z_3+\tau ^3 \tilde{y} z_2^2\\
			+\tau ^3 \tilde{y} z_2^2 z_3^2+\tau ^3 \tilde{y} z_3^2+3 \tau ^3 \tilde{y} z_2 z_3+\tilde{y} z_2 z_3-\tau ^3 z_2 z_3+z_2 z_3
		\end{array}
		\right) \\
		&\qquad\qquad \qquad \Bigg/ z_2z_3 \left(1-\frac{\tau^{\frac{3}{2}}}{z_2}\right)\left(1-\tau^{\frac{3}{2}}z_2\right)\left(1-\frac{\tau^{\frac{3}{2}}}{z_3}\right)\left(1-\tau^{\frac{3}{2}}z_3\right)
	\end{split}
\end{equation}
where we have chosen to grade by $D^I$, $\partial/\partial \phi_2$ and $\partial/\partial \phi_3$. We can check the lowest coefficient of $\tilde{y}$ yields:
\begin{equation}
	\begin{split}
		- \lim\limits_{\tilde{y} \rightarrow 0 } \tilde{y}^{3/2} \mathcal{Z} & = \frac{1-\tau ^3}{\left(1-\frac{\tau ^{3/2}}{z_2}\right) \left(1-\tau ^{3/2} z_2\right) \left(1-\frac{\tau ^{3/2}}{z_3}\right) \left(1-\tau ^{3/2} z_3\right)}\\
		& = 1+\tau ^{3/2} \left(z_2+z_3+\frac{1}{z_3}+\frac{1}{z_2}\right)+...\\
		&\quad+\tau ^3 \left(z_2^2+z_3 z_2+\frac{z_2}{z_3}+z_3^2+\frac{1}{z_3^2}+\frac{z_3}{z_2}+\frac{1}{z_3 z_2}+\frac{1}{z_2^2}+1\right)+...\\
		&\quad+\tau ^{9/2} \Bigg(z_2^3+z_3 z_2^2+\frac{z_2^2}{z_3}+z_3^2 z_2+\frac{z_2}{z_3^2}+z_2+z_3^3+z_3+\frac{1}{z_3}+\frac{1}{z_3^3}+\frac{z_3^2}{z_2}+ ...\\
		&\qquad \qquad \qquad +\frac{1}{z_3^2 z_2}+\frac{1}{z_2}+\frac{z_3}{z_2^2}+\frac{1}{z_3 z_2^2}+\frac{1}{z_2^3}\Bigg)+O\left(\tau ^{11/2}\right)
	\end{split}
\end{equation}
i.e. indeed a polynomial in positive powers of $\tau$ with coefficients positive (Laurent) polynomials in the other fugacities, as expected of a Hilbert series and so in agreement with (\ref{hilbertseries}). It is also easy to check that the highest coefficient of $\tilde{y}$ is:
\begin{equation}
	\lim_{\tilde{y}\rightarrow \infty} \tilde{y}^{-\frac{3}{2}} \mathcal{Z} = \tau^3  \left(- \lim\limits_{\tilde{y} \rightarrow 0 } \tilde{y}^{3/2} \mathcal{Z}\right)
\end{equation}
and therefore also has positive expansion in powers of $\tau$, verifying consistency check (\ref{consistencycheck2}). The monomial $\tau^{d_{\mathbb{C}}}$ corresponds to the character of the group action corresponding to fugacities $\{\tau, z_{2}, z_{3}\}$ on the nowhere-vanishing holomorphic $(3,0)$ form $\Omega$ described in section \ref{sectioncalabiyaucones}.\\

We note that the formula for the superconformal index for the conifold is reminiscent of the `single particle index' of the Klebanov-Witten 4d SCFT obtained in \cite{Gadde:2010en}.\footnote{We thank the anonymous referee for drawing our attention to this fact.} This is suggestive of a relation between the superconformal quantum mechanical index for Calabi-Yau 3-fold conical singularities and the four-dimensional superconformal index for the N = 1 SCFT arising from D3-branes probing such a singularity. We leave an exploration of this to future work.

\subsubsection{The $Y^{p,q}$ singularities} 
Here we consider the infinite family of toric Calabi-Yau 3-folds called the $Y^{p,q}$ singularities with $p$ and $q$ coprime, see e.g. \cite{sparksmartelliYpq, Ypqsingularities}. These have the following toric data:
\begin{equation}
	\{v_{\gamma}\} = \{(1,0,0), (1,1,0), (1,p,p), (1,p-q-1,p-q)\}
\end{equation}
The toric diagram is given in figure \ref{ypqdiagram}. The Reeb vector is given by \cite{sparksmartelliYpq, Ypqsingularities}: 
\begin{equation}\label{reebvectorypq}
	b = \left(3, \frac{1}{2}\left(3p-3q+l^{-1}\right),  \frac{1}{2}\left(3p-3q+l^{-1}\right)\right)
\end{equation}
where $l^{-1} = \frac{1}{q}\left(3q^2-2p^2+p\sqrt{4p^2-3q^2}\right)$.
\begin{figure}
	\centering
	\includegraphics[height=60mm,angle=0,  clip=true]{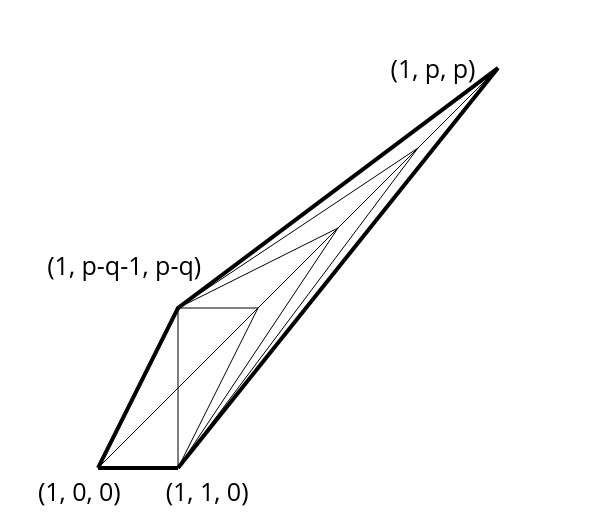}
	\caption{The toric diagram of the $Y^{p,q}$ singularities, together with a possible basic lattice triangulation by including interior points $\{(1,a,a)\}$ with $a=1,..,p-1$}\label{ypqdiagram}
\end{figure}
We denote by $ A_a$ the triangles in the triangulation with vertices $\{(1,a-1,a-1),(1,1,0),(1,a,a)\}$ and inward facing normals $\{(1-a,-1+a,2-a),(a,-a,-1+a),(0,1,-1)\}$. We denote by $B_{a}$ the triangles with vertices $\{(1,a-1,a-1), (1,a,a), (1,p-q-1,p-q)\}$, $a=1,...,p$ and inward facing normals $\{(0,-1,1), (a,a-p+q,-1-a+p-q), (1-a,1-a+p-q,a-p+q))\}$. It is then easy to write down the superconformal index as:
\begin{equation}
	\begin{split}
		\mathcal{Z} =& -\tilde{y}^{3/2} \sum_{a=1,..,p} \Bigg[\frac{\left(1-\tilde{y}\frac{x_2}{x_3}\right) \left(1-\tilde{y} x_1^{1-a} x_2^{a-1} x_3^{2-a}\right) \left(1-\tilde{y} x_1^a x_2^{-a} x_3^{a-1}\right)}{\left(1-\frac{x_2}{x_3}\right) \left(1-x_1^{1-a} x_2^{a-1} x_3^{2-a}\right) \left(1-x_1^a x_2^{-a} x_3^{a-1}\right)}\\
		& +\frac{\left(1- \tilde{y} \frac{x_3}{x_2}\right) \left(1-\tilde{y} x_1^a x_2^{a-p+q} x_3^{-a+p-q-1}\right) \left(1-\tilde{y} x_1^{1-a} x_2^{-a+p-q+1} x_3^{a-p+q}\right)}{\left(1-\frac{x_3}{x_2}\right) \left(1-x_1^a x_2^{a-p+q} x_3^{-a+p-q-1}\right) \left(1-x_1^{1-a} x_2^{-a+p-q+1} x_3^{a-p+q}\right)}\Bigg]
	\end{split}
\end{equation}
The first term in the summand is from fixed points corresponding to triangles $\{A_a\}$ and the second from $\{B_a\}$. To obtain the index in fugacity $\tau$, make the relabelling (\ref{fugacityrelabelling}) using the above Reeb vector. The resulting expression is rather grotesque so we omit it here. However we can easily see that: 
\begin{equation}
\lim_{\tilde{y}\rightarrow \infty} \tilde{y}^{-\frac{3}{2}} \mathcal{Z} = x_1  \left(- \lim\limits_{\tilde{y} \rightarrow 0 } \tilde{y}^{3/2} \mathcal{Z}\right)
\end{equation}
since the ratio of the highest power of $\tilde{y}$ and the lowest is $-x_1$ at each individual fixed point. Performing the rescaling of fugacities corresponding to (\ref{reebvectorypq}), we see we may factor out $\tau^3$, the character corresponding to the nowhere vanishing $(3,0)$-form.

\subsection{Quiver Varieties}\label{sectionquivervarieties}

\subsubsection{Nakajima Quiver Varieties}
A large class of examples to which our analysis applies are the Nakajima quiver varieties. In fact, it is possible to define on them a hyperK\"ahler metric specified by the hyperK\"ahler quotient construction, thus these varieties are hyperK\"ahler cones. As stated before, on such spaces the superconformal algebra is enlarged to $\mathcal{N} = (4,4)$ and the model was completely analysed in the papers \cite{doreybarns-graham, doreysingleton}. HyperK\"ahler spaces are $4k$-dimensional where $k$ is an integer. The hyperK\"ahler metric is Ricci-flat, and the cone is Calabi-Yau since there is a global holomorphic $(2k,0)$-form given by the $k$-th exterior power of the canonical holomorphic symplectic $(2,0)$-form $\omega_{\mathbb{C}}$. The superconformal index of K\"ahler cones coincides with the index defined for hyperK\"ahler cones in \cite{doreysingleton} when the cone is hyperK\"ahler. As stated before there is a manifest $y \rightarrow 1/y$ symmetry, where $y = \tilde{y}\tau$, induced via the action of the operator $\wedge \omega_{\mathbb{C}}$. This clearly implies the factorisation property (\ref{factorisationproperty}), but is obviously stronger. The reason we introduce them here is to consider instead a different Reeb vector in the Reeb cone, defining a K\"ahler metric which is not the hyperK\"ahler metric.\footnote{The hyperK\"ahler metric should coincide with the one found by the volume minimisation procedure of \cite{sparksmartelliyauvolume}, since it is Ricci-flat and Calabi-Yau and hence defines a Sasaki-Einstein manifold. The results of section \ref{sectioncalabiyaucones} therefore hold when the link is smooth and the singularity is isolated.} Although for a given Reeb vector we cannot prove the existence of a metric corresponding to it, assuming that it does exist an interesting analysis of the superconformal index can be done. \\

Note that for Nakajima quiver varieties, and also the handsaw quiver varieties defined in the following section, it is possible that there are in addition to the singularity at the tip of the cone, singular subspaces which intersect the tip. In this case the Sasakian base $Y$ is singular. We naturally extend the definition of the Reeb cone (as the subset of $\mathfrak{t}_s$ so that all non-constant holomorphic functions on the variety are graded positively under the action) to this case. We will see that the canonical Reeb vector induced from affine space in the construction of quiver varieties lies in this subset.\\

We briefly recap the definition of Nakajima quiver varieties. Here we give their construction as complex algebraic varieties, following for instance \cite{nakajimaquivervarietiesandfinitedimensionalreps}. The hyperK\"ahler quotient construction and their equivalence is also described therein. A Nakajima quiver variety is specified by a (double-arrow) quiver $\Gamma= (I,E)$, where $I$ is a set of vertices and $E$ a set of edges between them. Let $H$ be the set of pairs consisting of a pair and its orientation. For an oriented edge $h$, let $\bar{h}$ be the same edge with reversed orientation. Let $\text{in}(h)$ be the incoming vertex of $h$, and $\text{out}(h)$ be the outgoing one. Let $\Omega$ be an orientation of the graph, i.e. $\Omega \subset H$ such that $\Omega \cup \bar{\Omega} = H$,  $\Omega \cap \bar{\Omega} = \emptyset$. In addition we specify $k \in \mathbb{Z}_{>0}^{I}$ and $N \in \mathbb{Z}_{\geq0}^{I}$. Let $V_i \cong \mathbb{C}^{k_i}$, $V \cong \bigoplus_i V_i$ and $W_i \cong \mathbb{C}^{N_i}$, $W \cong \bigoplus_i W_i$.  Define the affine space:
\begin{equation}
\begin{split}
M\equiv M(k,N)&=\bigoplus_{(i,j)\in
	\Omega}\text{Hom}\left(V_i, V_j\right)\oplus\text{Hom}\left(V_j, V_i\right)\\
&\phantom{=}\oplus\bigoplus_{i\in
	I}\text{Hom}\left(W_i, V_i\right)\oplus
\text{Hom}\left(V_i, W_i\right)
\end{split}
\end{equation}
and a general element of this direct sum to be:
\begin{equation}
(B, \tilde{B},i,j) = \left(\{B_h\}_{h\in \Omega}, \{\tilde{B}_h\}_{h\in\Omega},\{i_k\}_{k \in I}, \{j_k\}_{k\in I}\right)
\end{equation} 
If we specify: 
\begin{equation}
\mu(B, \tilde{B},i,j) = [B, \tilde{B}]+ij \in \mbox{End}(V,V)
\end{equation}
$\mu^{-1}(0)$ defines a variety. Define $G = \prod\mbox{GL}(V_i)$ acting on $M$ by
\begin{equation}\label{quivergaugeaction}
g \in G : (B, \tilde{B},i,j) \mapsto (g^{-1} B g , g^{-1} \tilde{B} g, g^{-1}i, jg)
\end{equation}
which acts on $\mu^{-1}(0)$ since it commutes with $\mu$. A point $(B, \tilde{B},i,j) \in \mu^{-1}(0)$ is said to be \textit{stable} if whenever a graded subspace $S=\bigoplus_{i=I}^{I}S_i$ of V is invariant under $B$ and $\tilde{B}$, and contained in $\text{ker}\, j$, then $S=0$. Defining $\mu^{-1}(0)^s$ to be the set of stable points, we can define the Nakajima quiver varieties:
\begin{equation}
\mathscr{M}_0 = \mu^{-1}(0) /\!\!/ G \qquad \mathscr{M} = \mu^{-1}(0)^s/G
\end{equation}
where the former quotient is the affine GIT quotient, and the latter can be regarded as a symplectic quotient. $\mathscr{M}_0$ is a hyperK\"ahler cone, and $\mathscr{M}$ is its projective symplectic resolution, in that there is a natural projective morphism $\pi: \mathscr{M} \rightarrow \mathscr{M}_0$ such that the $\pi^* \omega$, the pullback of the symplectic form on the open set of smooth points of $\mathscr{M}_0$, can be extended to a symplectic form on all of $\mathscr{M}$. In the hyperK\"ahler quotient construction, $\mu$ is the complex moment map for the action of $G_{\mathbb{R}}=\prod\mbox{U}(V_i)$, the real form of $G$. There is an addition a real moment map $\mu_{\mathbb{R}}$. Then $\mathscr{M}_0 = \left(\mu^{-1}(0)\cap \mu_{\mathbb{R}}^{-1}(0)\right)/G_{\mathbb{R}}$, and $\mathscr{M} = \left(\mu^{-1}(0)\cap \mu_{\mathbb{R}}^{-1}(\zeta_{\mathbb{R}})\right)/G_{\mathbb{R}}$, where $\zeta_{\mathbb{R}} \in \mathbb{R}^I$ is generic. We call these $X$ and $\tilde{X}$ hereonafter for consistency.

\begin{theorem}\label{lusztigtheorem} (Lusztig \cite{lusztigonquivervarieties}): The coordinate ring of a general quiver variety is generated by the elements:
\begin{itemize}
	\item $\text{Tr}\left(C_{h_r}C_{h_{r-1}}...C_{h_1}\right)$ where $h_1, h_2...,h_r$ defines a \textit{cycle} in $I$, i.e. a sequence in $H$ such that $\text{out}(h_1) = \text{in}(h_2)$, $\text{out}(h_2) = \text{in}(h_3)$, ... , $\text{out}(h_r) = \text{in}(h_1)$, and 
	\begin{equation}
	C_{h} =
	\begin{cases}
	B_{h}	 & \text{if}\,\, h\in \Omega \\
	\tilde{B}_{h} & \text{if}\,\, h\in 
	\bar{\Omega}
	\end{cases}
	\end{equation}
	
	\item $\langle\eta, j_{\text{out}(h_r)}C_{h_r}C_{h_{r-1}}...C_{h_1}i_{\text{in}(h_1)}\rangle$ where $h_1, h_2...,h_r$ defines a \textit{path} in $I$, i.e. a sequence in $H$ such that $\text{out}(h_1) = \text{in}(h_2)$, $\text{out}(h_2) = \text{in}(h_3)$, ... , $\text{out}(h_{r\shortminus1}) = \text{in}(h_r)$, and $\eta$ is a linear form on $\text{Hom}\left(W_{\text{in}(h_1)}, W_{\text{out}(h_r)}\right)$.
\end{itemize}
\end{theorem}
\vspace*{5mm}
We now specialise to the \textit{linear} or $A$-type quivers, whose vertices $I$ (also known as gauge nodes in the physics literature) form the Dynkin diagram of the Lie algebra $A_{n}$. The arguments extend easily to the general Nakajima quiver variety. The quiver diagram is given in figure \ref{linearquiverdiagram}. Notice a slight change in the notation for $B, \tilde{B}$.
\begin{figure}
	\centering
	\includegraphics[height = 35mm,angle=0,  clip=true]{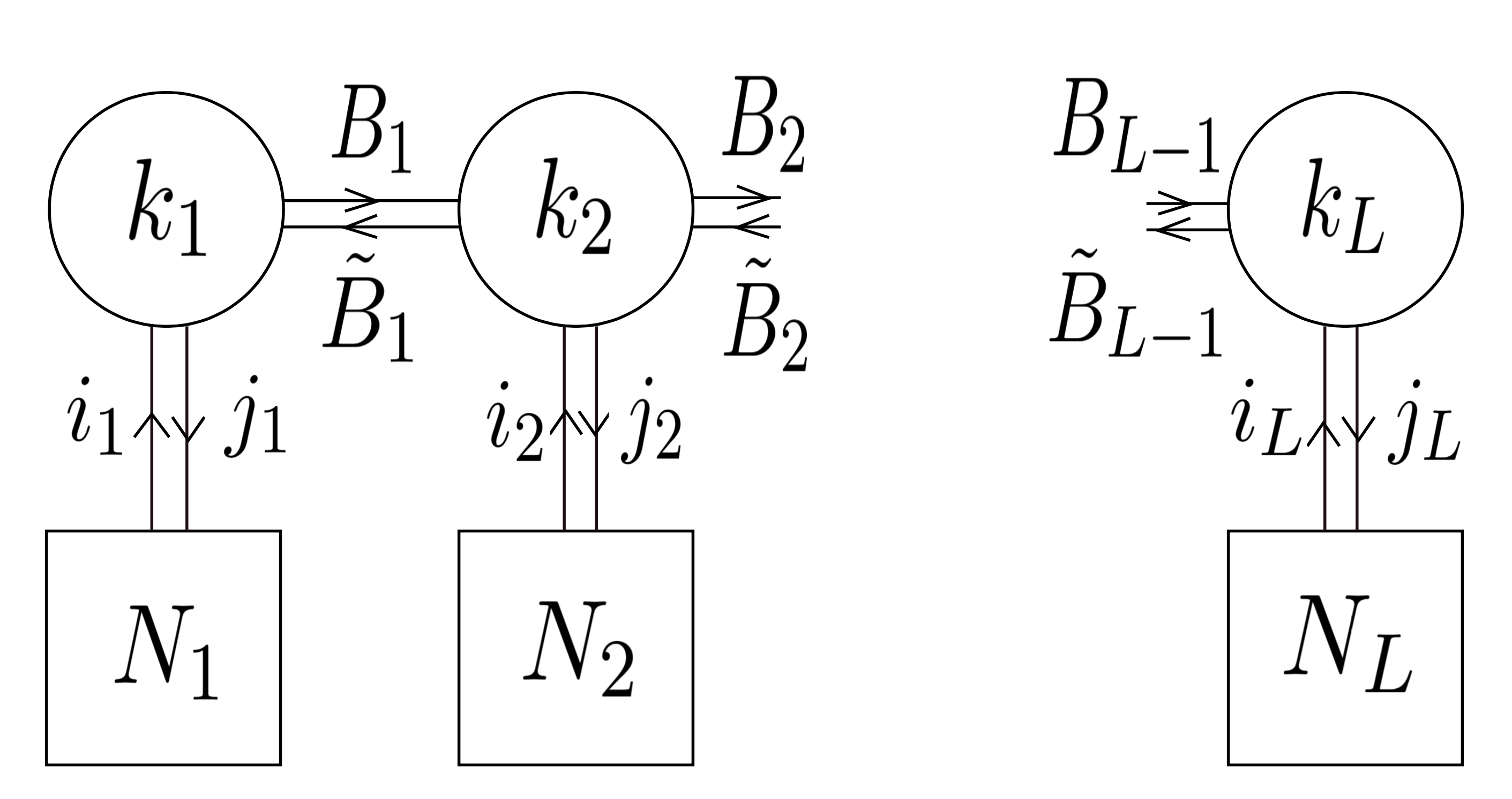}
	\caption{The quiver diagram for the general linear quiver. Double arrows have been used to indicate the components of the pre-quotient affine space. Note we are using the notation of Nakajima in \cite{nakajimahandsaw}. In the $4d$ $\mathcal{N}=2$ gauge theory notation, the double arrows would be replaced by single undirected edges. In the $4d$ $\mathcal{N}=1$ notation single arrows from each gauge node to itself corresponding to adjoint chirals should be included.}\label{linearquiverdiagram}
\end{figure}
The quiver variety has a natural action of $G_w \equiv \prod_{i=1}^{L}\text{GL}(W_i)$ in addition to the Reeb vector, this acts as:
\begin{equation}
(\tau, h) \in \mathbb{C}^* \times G_w\,\,,  \qquad (\tau, h)  \cdot (B, \tilde{B},i,j) \mapsto (\tau B, \tau \tilde{B}, \tau i h^{-1}, \tau h j)
\end{equation}
We take the torus of holomorphic isometries to be $T = \mathbb{C}^* \times \prod_{i=1}^{L}\mathbb{T}(W_i)$ where $\mathbb{T}(W_i)$ is the diagonal maximal torus of $\text{GL}(W_i)$. We introduce fugacities $\{z_{i}^{\alpha}\}$ where $i\in I$ and $\alpha = 1, ... ,N_i$. The linear quiver has isolated fixed points under $T$ and therefore its superconformal index can be computed using the localisation formula (\ref{scindexlocalisationformula}), see \cite{doreybarns-graham} for details.\\

We now proceed to characterise the Reeb cone for the general $A$-type quiver, specified by requiring that all non-constant functions have positive weights. The $\mathbb{C}^*$ action above corresponds to the action of the Reeb vector $D^I$ for the hyperK\"ahler metric obtained as a result of the hyperK\"ahler quotient construction of the quiver variety, and is induced from the canonical dilatation/Reeb vector on the pre-quotient affine space. The superconformal index corresponding to a different choice of Reeb vector will be given by a relabelling of fugacities. Let the new Reeb vector be given by:
\begin{equation}\label{newreebvector}
\tilde{D}^I =  \nu D^I + \lambda_{i}^a \frac{\partial}{\partial \phi_{i}^a}
\end{equation}
where $\{\frac{\partial}{\partial \phi_{i}^a}\}$, $a=1, ..., N_i$ generate $\mathbb{T}(W_i)$. Then the relabelling is:
\begin{equation}
\tau \mapsto \tilde{\tau}^{\nu} \qquad z_i^a \mapsto \tilde{\tau}^{\lambda_{i}^a}z_i^a
\end{equation}
where $\tilde{\tau}$ is the fugacity corresponding to the new Reeb vector. The Reeb cone will describe a cone in $\left( \nu, \{\lambda_{i}^a\}\right)$ which can be regarded as coordinates in $T$. Note that instead of the full $G_{w}$, the true action on the quiver variety is given by $G_{w}/\mathbb{C}^*$ where $\mathbb{C}^*$ is  the action on the linear data generated by the identity element in $G_{w}$, i.e. $h = \otimes_{i=1}^{L} \text{id}_{W_i}$, since this is actually a gauged out in the quotient (\ref{quivergaugeaction}). Thus to obtain the Reeb cone we should impose $\sum_{i}^{L} \sum_{a}^{N_i} \lambda_{i}^a = 0$, or equivalently choose a basis of $G_{w}/\mathbb{C}^*$.\\

Consider the generators of holomorphic functions in theorem \ref{lusztigtheorem}. $B$ and $\tilde{B}$ are charged only under the canonical Reeb $\mathbb{C}^*$, hence in order for the first type of holomorphic function to be positively graded under $\tilde{D}^I$ with fugacity $\tilde{\tau}$, we need $\nu>0$. The second type of function is specified by a choice of vertices $k = \text{in}(h_1)$ and $l = \text{out}(h_r)$ and a path between them $h_1, h_2, ..., h_r$.  The minimally charged generator is the one with the shortest path between vertices $k$ and $l$, since then there are fewer $B$ or $\tilde{B}$ to contribute positive powers of $\tilde{\tau}$. The minimally charged generator is, using the notation of figure \ref{linearquiverdiagram}:
\begin{equation}
F_{k,l}^{(\eta)} \equiv \begin{cases}
\langle\eta, j_{l}B_{l}B_{l-1}...B_k i_{k}\rangle	 & \text{if}\,\, l>k \\
\langle\eta, j_{l}\tilde{B}_{l}... \tilde{B}_{k-1} \tilde{B}_{k} i_{k}\rangle	 & \text{if}\,\, k>l  \\
\langle\eta, j_{k}i_{k}\rangle	 & \text{if}\,\, l=k
\end{cases}
\end{equation}
for $\eta$ a general linear form on $\text{Hom}\left(W_{k}, W_{l}\right)$.  The element of $\text{Hom}\left(W_{k}, W_{l}\right)$ appearing in the second argument of the above inner products is gauge invariant. By considering the transformation under $\tilde{D}^I$ of each of its entries, the following constraints are derived on $\left( \nu, \{\lambda_{i}^a\}\right)$. These are:
\begin{equation}
 \forall \,\,k,l \in I \,\, , a\in\{1,...,N_k\} \, ,b\in\{1,...,N_l\} : 
\begin{cases}
&\left(2 + |k-l|\right)\nu + \lambda_{k}^{a} - {\lambda_{l}}^{b} > 0 \\
&\left(2 + |k-l|\right)\nu + \lambda_{l}^{b} - \lambda_{k}^{a} > 0
\end{cases}
\end{equation}
These constraints define the Reeb cone, together with the constraint $\sum_{i}^{L} \sum_{a}^{N_i} \lambda_{i}^a = 0$. Notice that the constraints with $k=l$ already impose that $\nu>0$, the constraint from the first type of holomorphic function, so if there is only one (gauge) vertex we still have this condition.\\\\

\begin{example}
	We do the example of the quiver in figure \ref{TCPNDiagram} explicitly. 
	\begin{figure}
		\centering
		\includegraphics[height = 35mm,angle=0,  clip=true]{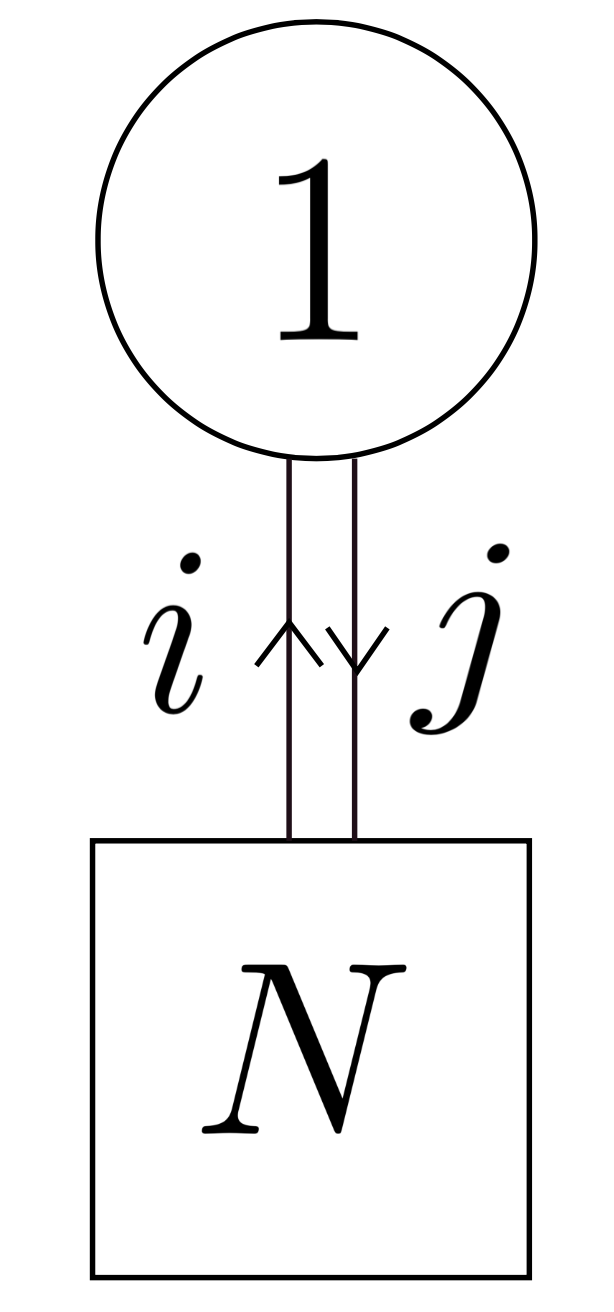}
		\caption{The quiver diagram for $T^*\mathbb{CP}^{N-1}$.}\label{TCPNDiagram}
	\end{figure}
	The resolved space $\tilde{X}$ is the cotangent bundle to $\mathbb{CP}^{N-1}$. The superconformal index was computed in \cite{doreybarns-graham}. Here we give an explicit description of the geometry, characterise the Reeb cone, and consider a different choice of Reeb vector. The linear data and moment map condition is specified by:
	\begin{equation}
	i = \left(i_1, i_2, ..., i_N\right)\,\,, \qquad j = 
	\begin{pmatrix}
	j_{1} \\
	j_{2} \\
	\vdots \\
	j_{N}
	\end{pmatrix}
	\,\,, \qquad \mu(i,j) =i \cdot j = 0
	\end{equation}
\end{example}
The stability condition implies that $j\neq 0$. Since the gauge action is $g \cdot (i,j) = (g^{-1}i, jg)$ where $g \in \mathbb{C}^*$, we can always choose $j$ to parametrise an element of $\mathbb{CP}^{N-1}$. $\mu=0$ is invariant under this and implies that $i$ defines a map:
\begin{equation}
i : \mathbb{C}^3\big/\text{Span}_{\mathbb{C}}(j) \rightarrow \mathbb{C}
\end{equation}
and therefore can be considered to lie in $T^{*}_{j}(\mathbb{CP}^{N-1})$. Thus the resolved space is $T^{*}(\mathbb{CP}^{N-1})$. To compute the index, consider the fixed points of $T = \mathbb{C}^* \times \mathbb{T}^N$ where the first factor corresponds to the action of the canonically induced Reeb vector, and the second factor is the maximal torus of the $\text{GL}(N,\mathbb{C})$. As stated before the actual action of the latter is $\mathbb{T}^{N-1}$ since the centre is gauged out. To be fixed by the $\mathbb{C}^*$, a generic point in $\tilde{X}$ whose representative can be chosen as $(i, j)$ where $j \in \mathbb{CP}^{N-1}$ must have $i=0$, since the canonical Reeb vector contracts the cotangent directions. $\mathbb{T}^{N}$ acts on an element of $\mathbb{CP}^{N-1}$ as: 
\begin{equation}
(z_1,z_2, ..., z_N) \cdot [j_1, j_2, ..., j_N] =   [z_1 j_1, z_2 j_2, ..., z_N j_N] 
\end{equation}
Here it is obvious from the definition of $\mathbb{CP}^{N-1}$ that the true action is only $\mathbb{T}^{N-1}$. The fixed points of $T^{*}(\mathbb{CP}^{N-1})$ under $T$ are the $N$ points:
\begin{equation}
\left(0, [1,0, ..., 0]\right) \quad,\quad \left(0, [0,1, ..., 0]\right), \quad ....\quad , \left(0, [0,0, ..., 1]\right) 
\end{equation}
i.e. they lie in $\mathbb{CP}^{N-1} \subset T^*\mathbb{CP}^{N-1}$. Given this, the superconformal index is easily computed:
\begin{equation}\label{scindexCPN-1}
\mathcal{Z}\left(T^* \mathbb{CP}^{N-1}\right) = \left(\frac{1}{\tilde{y}}\right)^{(N\shortminus 1)} \frac{\left(1 \shortminus \tilde{y}\frac{z_2}{z_1}\right)...\left(1\shortminus \tilde{y}\frac{z_N}{z_1}\right)}{\left(1\shortminus \frac{z_2}{z_1}\right)...\left(1\shortminus \frac{z_N}{z_1}\right)} \cdot \frac{\left(1\shortminus \tilde{y}\tau^2\frac{z_1}{z_2}\right)...\left(1\shortminus \tilde{y}\tau^2\frac{z_1}{z_N}\right)}{\left(1\shortminus \tau^2\frac{z_1}{z_2}\right)...\left(1\shortminus \tau^2\frac{z_1}{z_N}\right)} +...
\end{equation}
Where the "$+...$" consist of terms identical to the one above with the role of $z_1$ switched for $z_2, ...., z_N$. The contribution displayed comes from the fixed point $\left(0, [1,0, ..., 0]\right) $. Note the index exhibits the Weyl invariance discussed explicitly in \cite{doreybarns-graham}, and is only dependent on the ratios of $Z$-fugacities as expected. The first terms come from the directions in $T^*\mathbb{CP}^{N-1}$ along the base, and the second terms from the cotangent fibres.\\

To be even more explicit, consider the example $N=3$. The Reeb cone is given by the constraints (dropping the lower  index on $\lambda$): $2\nu + \lambda^a- \lambda^b>0$, $\forall \,a,b=1,2,3$. In fact there are only 4 independent constraints, reflecting the fact there is really only an action of $\mathbb{C} \times \mathbb{T}^2$. Relabelling $m= \lambda_1 - \lambda_2$ and $n  = \lambda_2 - \lambda_3$ to reflect the true independent rescalings of the fugacities/choice of Reeb vector, the Reeb cone is specified in $(\nu,m,n)$ space by:
\begin{equation}
2\nu+m>0,\quad 2\nu-m>0,\quad 2\nu+n>0,\quad 2\nu-n>0
\end{equation}
This is convex rational polyhedral cone with edge vectors $(1,2,2)$, $(1,2,-2)$, $(1,-2,2)$, $(1,-2,-2)$. \\

Consider now the fixed point submanifolds of $T^*\mathbb{CP}^2$ corresponding to 2 different choices of Reeb vector in the Reeb cone, as a verification of the results of section \ref{sectionreeblimit}. Let $D^I$ correspond to the canonical Reeb vector induced by the hyperK\"ahler quotient. The fixed point subvariety of this is simply $\mathbb{CP}^2$. Taking the $\tau \rightarrow 0$ limit of (\ref{scindexCPN-1}), we obtain:
\begin{equation}
\lim_{\tau \rightarrow 0 } \mathcal{Z}\left(T^* \mathbb{CP}^{2}\right)  = \frac{1}{\tilde{y}^{2}} \left(1+\tilde{y}+\tilde{y}^2\right)
\end{equation}
which is exactly $1/\tilde{y}^2$ times the Poincar\'e polynomial of $\mathbb{CP}^{2}$ (explicitly it is $P_{\sqrt{\tilde{y}}}(\mathbb{CP}^2)$) and is independent of $Z$. Now suppose we choose a different Reeb vector $\tilde{D}^I$, specified by $\nu=1$ and $\lambda^1=1$ in (\ref{newreebvector}). Now the fixed submanifold of $\tilde{D}^I$ in $T^*(\mathbb{CP}^2)$ consists of the union of an isolated point $[1,0,0]$, and a $\mathbb{CP}^1$ (given by points of the form $[0,j_2, j_3]$) lying in  $\mathbb{CP}^{2}$. Rescaling fugacities and taking the $\tilde{\tau}\rightarrow 0$ limit of the index:
\begin{equation}
\lim_{\tilde{\tau} \rightarrow 0 } \mathcal{Z}\left(T^* \mathbb{CP}^{2}\right)  = \frac{1}{\tilde{y}^{2}} \left( \tilde{y}^2+ \frac{1-\tilde{y}\frac{z_3}{z_2}}{1-\frac{z_3}{z_2}} + \frac{1-\tilde{y}\frac{z_2}{z_3}}{1-\frac{z_2}{z_3}}\right) = \frac{1}{\tilde{y}^{2}} \left(\tilde{y}^2 + (1+\tilde{y})\right)
\end{equation}
Here we have been explicit in the first equality with the contributions of the various fixed points. The $\tilde{y}^2$ contribution comes from the Poincar\'e polynomial of the single point $[1,0,0]$ on $\mathbb{CP}^2$. The power of $\tilde{y}$ comes from the fact that at the isolated fixed point, there are two (co)tangent directions in $T^*\mathbb{CP}^2$ negatively charged under the action of $\tilde{D}^I$, i.e. those along the $\mathbb{CP}^2$. The remaining two summands come from the north and south pole of $\mathbb{CP}^1 \cong S^2$. These have no negatively-charged direction in the normal bundle. They combine to give $1+\tilde{y} = P_{\sqrt{\tilde{y}}}(\mathbb{CP}^1)$. Both contributions sum to give the same result as for the original Reeb vector as claimed.\\

Note that in general, taking the fixed point subvariety of a $\mathbb{C}^* \subset T$ action on the unresolved hyperK\"ahler cone $X$ lying outside the Reeb cone, we obtain a non-compact disjoint union of K\"ahler cones. To see this, note that any  $\mathbb{C}^* \subset T$ commutes with the canonical Reeb vector. Therefore the action of the canonical Reeb vector, hence the associated dilatation, is defined on any fixed subvariety. The subvarieties are K\"ahler since the action of $T$ is holomorphic with respect to at least one of three complex structures defined on the hyperK\"ahler cone. These subvarieties have a resolution induced from the resolution of the original quiver variety. We can see this explicitly in the example of the handsaw quiver varieties considered in the following section, which can be obtained as fixed point subvarieties of the ADHM quiver variety \cite{nakajimahandsaw}. 

\subsubsection{The Handsaw Quiver Variety}\label{sectionhandsaw}

Another large class of K\"ahler cones is provided by the handsaw quiver varieties of Nakajima \cite{nakajimahandsaw}. They are single-arrow quivers, and the construction proceeds similarly to the double-arrow quivers above. They are K\"ahler but not necessarily hyperK\"ahler. They describe the moduli space of vortices of the  $T_{\rho}(SU(N))$ 3d $\mathcal{N}=4$ triangular quiver gauge theories. The latter is specified in figure \ref{handsawtrhodiagram}, where $L\in \mathbb{Z}_{>0}$, $N_L =N$ and we assume $N_a > N_{a-1}$ for $a =  1,...,L$. Let $\rho_a = N_a-N_{a-1}$ so that $\rho = [\rho_1,...,\rho_L]$ forms a partition of $N$. The moduli space of vortices with vortex number $\mathfrak{n} = \{\mathfrak{n}_1,...,\mathfrak{n}_{L-1}\}$ with respect to each of the gauge groups in the 3d theory, is specified by the handsaw quiver variety also given in figure \ref{handsawtrhodiagram}. Denote by $V_b \equiv \mathbb{C}^{\mathfrak{n}_b}$, $b=1,..,L-1$, the vector space corresponding to the $\mathfrak{n}_b$ gauge node. Denote by $W_a \equiv \mathbb{C}^{\rho_a}$ the vector space corresponding to the $\rho_a$ flavour node. Also let $V = \bigoplus_{i=1}^{L-1}V_i$, $W = \bigoplus_{i=1}^{L}W_i$. Thus $\mathfrak{n}$ and $\rho$ specify the handsaw, and are dimension vectors. For details, we refer the reader to \cite{vorticesandvermas}. \\
\begin{figure}
	\begin{minipage}{0.5\textwidth}
		\centering
		\includegraphics[height=40mm,angle=0,  clip=true]{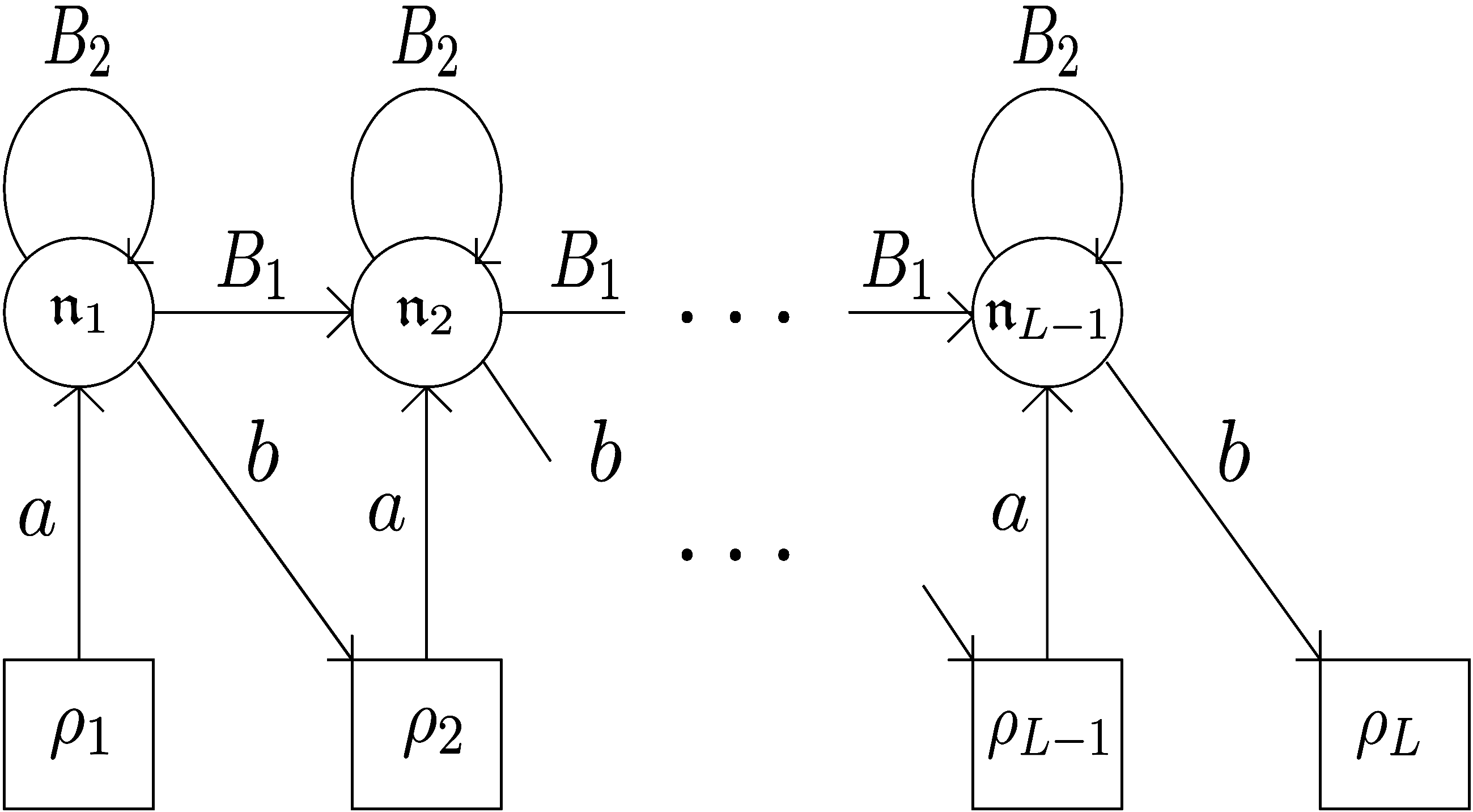} 
	\end{minipage}\hfill
	\begin{minipage}{0.5\textwidth}
		\centering
		\includegraphics[height=33mm,angle=0,  clip=true]{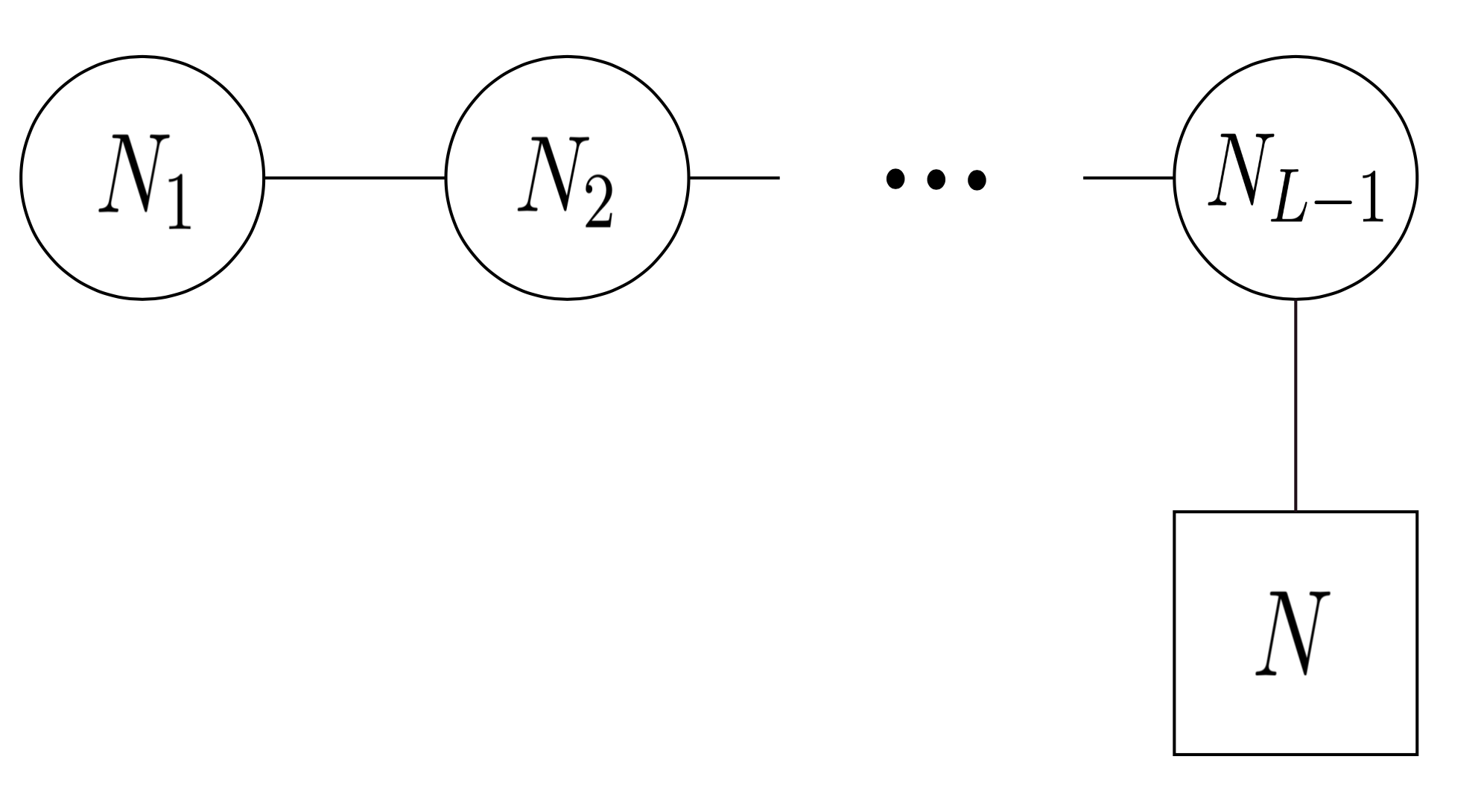} 
	\end{minipage}
	\caption{The handsaw quiver diagram, and the $T_{\rho}(SU(N))$ quiver diagram}\label{handsawtrhodiagram}
\end{figure}

Define:
\begin{equation}
\begin{split}
B_1 &\in \bigoplus_{i=1}^{L-2} \mbox{Hom}\left(V_{i}, V_{i+1}\right) \\ 
a &\in \bigoplus_{i=1}^{L-1} \mbox{Hom}\left(W_{i}, V_{i}\right)
\end{split}
\qquad
\begin{split}
B_2 &\in \bigoplus_{i=1}^{L-1} \mbox{End}\left(V_{i} \right) \\
b &\in \bigoplus_{i=1}^{L-1} \mbox{Hom}\left(V_{i}, W_{i+1}\right)
\end{split}
\end{equation}
and consider the affine space of all quadruples $(B_1,B_2,a,b)$ of linear data. Defining:
\begin{equation}
\mu(B_1, B_2, a, b) = [B_1, B_2]+ab \in \mbox{End}(V,V)
\end{equation}
then $\mu^{-1}(0)$ specifies an affine variety with a natural group action of $G = \prod\mbox{GL}(V_i)$ given by its action on the linear data as: 
\begin{equation}
g \in G : (B_1, B_2, a, b) \mapsto (g^{-1} B_1 g , g^{-1} B_2 g, g^{-1}a, bg)
\end{equation}
A point in the space of linear data $(B_1, B_2, a, b)$ is called \textit{stable} if there is no proper graded subspace $S=\bigoplus_{i=1}^{L-1}S_i$ of $V$ stable under $B_1$, $B_2$ and containing $a(W)$, and \textit{costable} if there is no non-zero graded subspace stable under $B_1$, $B_2$  and contained in $\mbox{ker}(b)$. Nakajima defines the handsaw quiver varieties as:
\begin{equation}
\begin{split}
\mathscr{L} &= \{(B_1, B_2, a, b) \in \mu^{-1}(0) \,|\,\text{stable}\}/ G \\
\mathscr{L}_0 &= \{(B_1, B_2, a, b) \in \mu^{-1}(0)\} /\!\!/ G \\
\mathscr{L}_0^{\text{reg}} &= \{(B_1, B_2, a, b) \in \mu^{-1}(0) \,|\, \text{stable and costable}\}/ G \\
\end{split}
\end{equation}
Here $/\!\!/$ denotes the affine GIT quotient. There is a projective morphism $\pi : \mathscr{L} \mapsto \mathscr{L}_0$, such that $\pi$ is an isomorphism between $\mathscr{L}_0^{\text{reg}}$ considered as an open subscheme in both $\mathscr{L}$ and  $\mathscr{L}_0$. There is an alternative description of $\mathscr{L}$, and $\mathscr{L}_0$. They can be described as K\"ahler quotients of the affine variety $\mu^{-1}(0)$ by the group $\prod\mbox{U}(V_i)$ whose complexification is $G$. $\mathscr{L}$ is obtained by taking the K\"ahler quotient at a non-zero level set, corresponding to FI parameters turned on in the 1D GLSM corresponding to the handsaw quiver, see  \cite{vorticesandvermas}. $\mathscr{L}_0$  is obtained by taking the K\"ahler quotient at the zero level set.\\

We will henceforth call $\mathscr{L}_0$ \textit{the} handsaw quiver variety. It is singular, and a K\"ahler cone by the argument at the end of the last subsection. To elaborate, Nakajima describes the handsaw quiver variety as the fixed point submanifold of a $\mathbb{C}^{\times}$-action on the ADHM quiver variety, where the $\mathbb{C}^{\times} \subset (\mathbb{C}^{\times})^2 \times (\mathbb{C}^{\times})^N  $ where the latter is the complexified maximal torus of an isometric group action on the ADHM quiver. This $\mathbb{C}^*$ action is given by (on the linear data for the ADHM variety):
\begin{equation}\label{handsawasfixedpointofadhm}
t_1 \in \mathbb{C}^* : (B_1, B_2, a, b) \mapsto (t_1B_1, B_2, a\rho_w(t_1)^{-1} ,t_1\rho_w(t_1) b) 
\end{equation}
where $\rho_w(t_1)$ acts by $t_1^{i}$ on $W_i$. The unresolved ADHM quiver variety is a hyperK\"ahler cone, whose superconformal index was calculated in \cite{doreybarns-graham}. The $\mathbb{C}^{\times}$-action is the complexification of a $\mathfrak{u}(1)$ isometry which is holomorphic with respect to one of the three complex structures on the ADHM quiver, and therefore its fixed point subvariety is K\"ahler (possibly singular). Note that a group action and its complexification with respect to one of the complex structures have the same fixed points. The $\mathbb{C}^{\times}$-action also commutes with the homothety, and all such vectors are of the form:
\begin{equation}
	V = c(x_j) r \frac{\partial}{\partial r} + V^i(x_j) \frac{\partial}{\partial x_i}
\end{equation}
where $x_i$ are the remaining coordinates. The zero set of this vector field away from the singularity is therefore given by $c(x_i) = X^i(x_i)=0$ and is thus a cone with a homothetic action. This is consistent from the point of view of the algebra as the commuting subalgebra of the $\mathbb{C}^{\times}$-action whose fixed submanifold we take in $\mathfrak{osp}(4^*|4)$ is precisely our $\mathfrak{u}(1,1|2)$. The regular/quasi-regular action of the Reeb vector on the ADHM quiver variety descends to a regular/quasi-regular action on the handsaw. We later see this manifested in the characters of the group action at each fixed point. \\

Thus, providing $\mathscr{L}_0^{\text{reg}}$ is non-empty, since it is a Zariski open subset of $\mathscr{L}$ and $\mathscr{L}_0$, $\pi$ is a birational morphism and provides a canonical resolution of singularities. It also implies that Grauert-Riemenschneider vanishing \cite{grauertriemenschneider} applies, so that we have a consistent regularisation of the index. Additionally, the requirement that $\mathscr{L}_0^{\text{reg}}$ is non-empty has a nice physical interpretation. $\mathscr{L}_0^{\text{reg}}$ corresponds to the moduli space of genuine vortices, i.e. those which are not point-like. One can see this from the stratification in \cite{nakajimahandsaw}:
\begin{equation}
\mathscr{L}_0(\mathfrak{n}) = \bigsqcup \mathscr{L}_0^{\text{reg}}(\mathfrak{n}-\mathfrak{n}') \times \text{Sym}^{\mathfrak{n}'_1} \mathbb{C} ... \times \text{Sym}^{\mathfrak{n}'_{L-1}} \mathbb{C}
\end{equation}
where $\mathfrak{n}' = \{\mathfrak{n}'_1,...,\mathfrak{n}'_{L-1}\}$ is such that $\mathfrak{n}'_i \leq \mathfrak{n}_i$. This is analogous to the story for instantons, see \cite{nakajimainstantoncountignonblowup1}. Cases where $\mathscr{L}_0^{\text{reg}}$ is non-empty include when $\rho_1 \leq \rho_2 \leq ... \leq \rho_L$ \cite{nakajimahandsaw}. This includes the case when $\rho_1 = ... = \rho_L =1$, the vortex moduli space for the $T(SU(N))$ gauge theory. \\

We now proceed to computing the superconformal index of the handsaw quiver variety $\mathscr{L}_0$ via resolving to $\mathscr{L}$. There is an action of $G_w \equiv \prod_{i=1}^{L}\text{GL}(W_i)$, which acts on the linear data by conjugation, and there is an additional $\mathbb{C}^*$ action given by:
\begin{equation}
(B_1, B_2, a, b) \mapsto (B_1, t B_2, a,  t b)
\end{equation}
which can be seen as inherited from an action on the ADHM variety. These commute with equation $\mu=0$ and the $G$-action and therefore descend to actions on the quotients. Following Nakajima, we fix a decomposition  into 1-dimensional subspaces: $W_i = \bigoplus_\alpha W_i^{\alpha}$ such that $\alpha = 1,..., \text{dim}W_i$. We consider the restriction $G_w$ to the torus $\prod_i \mathbb{T}_i$, $\mathbb{T}_i \subset \text{GL}(W_i)$ preserving these subspaces. We also include the $\mathbb{C}^*$ action to form a larger decomposition-preserving torus $\mathbb{T}_w \equiv \mathbb{C}^* \times \prod_i \mathbb{T}_i$. The $\mathbb{T}_w$ fixed points on $\mathscr{L}$ correspond to tuples of Young diagrams $\Vec{Y} = \{Y_i^{\alpha}\}$ (French notation) corresponding to each $W_i^{\alpha}$ where the bottom-left corner of $Y_i^{\alpha}$ is shifted such that its $x$-coordinate is $i$, with the restriction that the total number of boxes in the tuple with $x=j$ is $\text{dim}V_j$. \\

Nakajima computed the character of the tangent space (equivalently cotangent space depending on convention) to be:
\begin{equation}\label{handsawfixedpointcharacter}
\mbox{ch}T_{\Vec{Y}}^{*}\mathscr{L} = \sum_{(i,\alpha),(j,\beta)} e_j^{\beta}(e_i^{\alpha})^{-1} \Bigg(\quad\,\,\, \sum_{\mathclap{\substack{s\in Y_i^{\alpha}\\
			l_{Y_j^{\beta}}(s)=0}}} t^{a_{Y_i^{\alpha}+1}} + \sum_{\mathclap{\substack{s\in Y_j^{\beta}\\
			l_{Y_i^{\alpha}}(s)=-1}}} t^{-a_{Y_j^{\beta}}}\Bigg)
\end{equation}
where $e_i^{\alpha}$ the fugacity corresponding to the action of the sub-torus of $\mathbb{T}_i$ acting on $W_i^{\alpha}$. The leg-length $l_{Y_{j}^{\beta}}(s)$ of a box $s \in Y_{i}^{\alpha}$ relative to $Y_{j}^\beta \in \Vec{Y}$ is the difference in $x$-coordinate of the right-most box in $Y_{j}^\beta$ in the same row as $s$, minus the $x$-coordinate of $s$. If there are no boxes in the same row, take the difference between $j-1$ and the $x$-coordinate of $s$. The arm-length $a_{Y_{i}^{\alpha}}(s)$ is the difference in $y$-coordinate of the top box in $Y_{i}^{\alpha}$ in the same column as $s$, and the $y$-coordinate of $s$.\\

A rescaling of fugacities must be performed to obtain the character in terms of the action of the Reeb vector, inherited from the ADHM variety (or equivalently from the flat space metric on the linear data). The Reeb vector acts as:
\begin{equation}\label{holoisomhandsaw}
(B_1, B_2, a, b) \rightarrow (\tau B_1, \tau B_2, \tau a, \tau b) \sim (B_1, \tau B_2,  a\rho_w(\tau), \tau \rho_w(\tau)^{-1} b)
\end{equation}
where the last gauge equivalence, which is an equality up identifying gauge orbits in the quotient, is given by (\ref{handsawasfixedpointofadhm}). Therefore we need to rescale:
\begin{equation}
\tau=t \qquad z_i^{\alpha} = e_i^{\alpha}t^{-i}
\end{equation}
so that now $\tau$ is the fugacity for the action of the Reeb vector, and $z_i^{\alpha}$ for the $\mathbb{C}^*$ acting on $W_i^{\alpha}$. Thus:
\begin{equation}
\mbox{ch}T_{\Vec{Y}}^*\mathscr{L} = \sum_{(i,\alpha),(j,\beta)} z_j^{\beta}(z_i^{\alpha})^{-1} \Bigg(\quad\,\,\, \sum_{\mathclap{\substack{s\in Y_i^{\alpha}\\
			l_{Y_j^{\beta}}(s)=0}}} \tau^{a_{Y_i^{\alpha}}+j-i+1} + \sum_{\mathclap{\substack{s\in Y_j^{\beta}\\
			l_{Y_i^{\alpha}}(s)=-1}}} \tau^{-a_{Y_j^{\beta}}-i+j}\Bigg)
\end{equation}
The formula for the index of a general handsaw, using (\ref{scindexlocalisationformula}) is then given by:
\begin{equation}\label{scindexhandsawquiver}
\begin{split}
\mathcal{Z} = (-)^{d_{\mathbb{C}}}\left(\frac{1}{\tilde{y}}\right)^{\frac{d_{\mathbb{C}}}{2}}\sum_{\mathclap{\substack{\Vec{Y}}}}\prod_{(i,\alpha),(j,\beta)}\quad
&\prod_{\mathclap{\substack{s\in Y_i^{\alpha}\\ l_{Y_j^{\beta}}(s)=0}}}PE\left(\left(1-\tilde{y}\right)\frac{z_j^{\beta}}{z_i^{\alpha}} \tau^{a_{Y_i^{\alpha}}+j-i+1}\right)\\
\times\quad&\prod_{\mathclap{\substack{s\in Y_i^{\alpha}\\ l_{Y_j^{\beta}}(s)=-1}}}PE\left(\left(1-\tilde{y}\right)\frac{z_i^{\alpha}}{z_j^{\beta}} \tau^{-a_{Y_i^{\alpha}}+i-j}\right)
\end{split}
\end{equation}
Where the sum over tuples of Young diagrams is over those such that the total number of boxes in the $i^{th}$ column is $\mbox{dim}V_i$. Note that it is possible to obtain this as a limit of the superconformal index on the ADHM quiver variety, via the same technique in \cite{doreybarns-graham} to obtain the index of the $A$-type quiver variety.\\

Note that not all boxes in the Young diagram contribute to the index, and that for each contribution corresponding to an individual fixed point in (\ref{scindexlocalisationformula}), the highest power of $\tilde{y}$ which appears should equal $d_{\mathbb{C}} = \sum_{i=1}^{L-1} \text{dim}(V_i)\left(\text{dim}(W_i)+ \text{dim}(W_{i+1})\right) = \sum_{i=1}^{L-1} \mathfrak{n}_i\left(\rho_i+ 
\rho_{i+1}\right)$ . This provides a non-trivial combinatorial identity condition on the Young diagrams.\\

We note that the vortex partition function of the $T_{\rho}(SU(N))$ quiver gauge theory is generated by superconformal indices of handsaw quiver varieties, its vortex moduli spaces. This will appear in a future work with Samuel Crew \cite{crewdoreyzhang}.

\section*{Acknowledgements}
The authors would like to thank Samuel Crew and Alec Barns-Graham for many helpful discussions. This work has been partially supported by STFC consolidated grant ST/P000681/1. 

\appendix
\section{The $\mathfrak{u}(1,1|2)$ Algebra}\label{appendixA}
In this appendix we list the generators and commutation relations for the superconformal $\mathfrak{u}(1,1|2)$ algebra. For the generators expressed in terms of conjugate momenta, see \cite{andrewthesis}. In the following, the bidegree of a form is $(p,q)$ and bosonic generators are self-adjoint.\\
The Hilbert space is:
\begin{equation}
\mathcal{H} = \Omega^*(X;\mathbb{C})
\end{equation}
the exterior algebra on $X$ under the usual $L^2$ norm. The bosonic subalgebra is:
\begin{equation}
	\mathfrak{g}_{B} = \mathfrak{su}(1,1) \oplus \mathfrak{su}(2) \oplus \mathfrak{u}(1)_{R^I} \oplus \mathfrak{u}(1)_{D^I}
\end{equation}
and the generators are (here $\omega$ and $K$ are the K\"ahler form and potential respectively):
\begin{equation*}
\begin{split}
\mathbb{H} &= \frac{1}{2}\Delta\\
J_3 &= \frac{1}{2}\left(p+q-d_{\mathbb{C}}\right)
\end{split}
\qquad
\begin{split}
\mathbb{D} &= - i \mathcal{L}_{D} + i \left(p+q-d_{\mathbb{C}}\right)\\
J_+ &= \omega \wedge
\end{split}
\qquad
\begin{split}
\mathbb{K} &= \frac{1}{2}\norm{D}^2\\
J_- &= (\omega \wedge)^{\dagger}
\end{split}
\end{equation*}
\begin{equation}
R^I = \frac{1}{2}(p-q) \qquad D^I = -i \mathcal{L}_{D^I}
\end{equation}
\begin{equation*}
\begin{split}
Q &= d\\
Q^I&= i\left(\bar{\partial}-\partial\right)
\end{split}
\qquad
\begin{split}
Q^{\dagger} &= d^{\dagger}\\
{Q^I}^{\dagger}&= i\left(\partial^{\dagger}-\bar{\partial}^{\dagger}\right)
\end{split}
\qquad
\begin{split}
S &=i dK\wedge\\
S^I &= \left(\partial - \bar{\partial}\right) K
\end{split}
\qquad
\begin{split}
S^\dagger &= -i\,\iota_{D}\\
{S^I}^{\dagger} &= -i \, \iota_{D^I}
\end{split}
\end{equation*}
and commutation relations (excluding those obtained by conjugation):
\begin{equation*}
\begin{split}
\left[\mathbb{D}, \mathbb{H}\right] &= 2i \mathbb{H}
\end{split}
\qquad
\begin{split}
\left[\mathbb{D}, \mathbb{K}\right] = -2i \mathbb{K}
\end{split}
\qquad
\begin{split}
\left[\mathbb{H}, \mathbb{K}\right] = -i \mathbb{D}
\end{split}
\end{equation*}
\begin{equation*}
\left[J_{3}, J_{\pm}\right]= \pm J_{\pm} \qquad \left[J_+, J_{-}\right]= 2J_3
\end{equation*}
\begin{equation}
\begin{split}
\left[\mathbb{D}, Q\right] &= i Q\\
\left[\mathbb{D}, Q^I\right] &= i Q^I\\
\left[J_3, Q \right] &= \frac{1}{2}Q\\
\left[J_+, Q^\dagger \right]&= -Q^I\\
\left[R^I, Q \right] &= \frac{i}{2} Q^I
\end{split}
\qquad
\begin{split}
\left[\mathbb{D}, S\right] &= -i S\\
\left[\mathbb{D}, S^I\right] &= -i S^I\\
\left[J_3, S \right] &= \frac{1}{2}S\\
\left[J_+, {Q^{I}}^{\dagger} \right]&= Q\\
\left[R^I, Q^I \right] &= -\frac{i}{2} Q
\end{split}
\qquad
\begin{split}
\left[\mathbb{H}, S\right] &= -i Q\\
\left[\mathbb{H}, S^I\right] &= -i Q^I\\
\left[J_3, Q^I \right] &= \frac{1}{2}Q^I\\
\left[J_+, S^\dagger \right]&= -S^I\\
\left[R^I, S \right] &= \frac{i}{2} S^I
\end{split}
\qquad
\begin{split}
\left[\mathbb{K}, Q\right] &= i S\\
\left[\mathbb{K}, Q^I\right] &= i S^I\\
\left[J_3, S^I \right] &= \frac{1}{2}S^I\\
\left[J_+, {S^{I}}^{\dagger} \right]&= S\\
\left[R^I, S^I \right] &= -\frac{i}{2} S
\end{split}
\end{equation}
\begin{equation*}
\begin{split}
\left\{Q, Q^{\dagger} \right\} &= 2\mathbb{H}\\
\left\{Q^I, {Q^I}^{\dagger} \right\} &= 2\mathbb{H}
\end{split}
\qquad
\begin{split}
\left\{S, S^{\dagger} \right\} &= 2\mathbb{K}\\
\left\{S^I, {S^I}^{\dagger} \right\} &= 2\mathbb{K}
\end{split}
\qquad
\begin{split}
\left\{Q, S^{\dagger} \right\} &= \mathbb{D}-2iJ_3\\
\left\{Q^I, {S^I}^{\dagger} \right\} &= \mathbb{D}-2iJ_3
\end{split}
\end{equation*}
\begin{equation*}
\left\{Q, S^I\right\} = -2iJ_+ \qquad \left\{Q^I, S \right\} = 2iJ_+ \qquad \left\{ Q, {S^I}^{\dagger }\right\} =  D^I \qquad \left\{ Q^I, S^{\dagger }\right\} =  - D^I 
\end{equation*}
\section{The Superconformal Index on Affine $\mathbb{C}^n$}\label{appendixB}

In this appendix we construct the superconformal index for flat space. We begin with the result for $\mathbb{C}$, where (using indices $m=1,2$ with no distinction between upper and lower), the generators of the conformal subalgebra are:
\begin{equation}
	\mathbb{H} = \frac{1}{2}P^2 \qquad \mathbb{K} = \frac{1}{2}X^2 \qquad \mathbb{D} = X_{m}P_{m}-i
\end{equation}
where $P_m = \dot{X}^m$, the homothety is given by $D^{m} = X^{m}$, and we use the standard complex structure on $\mathbb{C}$.
Define: 
\begin{equation}
	a_m = \sqrt{\frac{1}{2}}\left(X_{m} + iP_{m}\right)
\end{equation}
The generators of $\mathfrak{su}(1,1)$ obtained via the basis change can be expressed as:
\begin{equation}
	\mathbb{L}_{0} = a^{\dagger}_{m}a_{m} + 1 \qquad \mathbb{L}_+ = -\frac{1}{2}a^{\dagger}_{m}a^{\dagger}_{m} \qquad  \mathbb{L}_- = -\frac{1}{2}a_{m}a_{m}
\end{equation}
A basis of supercharges which are eigenvalues of $\mathbb{L}_0$ can be expressed in terms of oscillators as:
\begin{equation}
	\begin{split}
		\mathcal{Q} \,\,&= \sqrt{2}\psi^{\dagger}_{m}a^{\dagger}_{m}  \\
		\mathcal{Q}^{I} &= \sqrt{2}\psi^{\dagger}_{m}I_{mn}a^{\dagger}_{n}
	\end{split}
	\qquad
	\begin{split}
		\tilde{\mathcal{Q}} \,\,&= \sqrt{2}\psi_{m}a^{\dagger}_{m}  \\
		\tilde{\mathcal{Q}}^{I} &= \sqrt{2}\psi_{m}I_{mn}a^{\dagger}_{n}
	\end{split}
	\quad
	\begin{split}
		\mathcal{S} \,\,&= \sqrt{2}\psi^{\dagger}_{m}a_{m}  \\
		\mathcal{S}^{I} &= \sqrt{2}\psi^{\dagger}_{m}I_{mn}a_{n}
	\end{split}
	\qquad
	\begin{split}
		\tilde{\mathcal{S}} \,\,&= \sqrt{2}\psi_{m}a_{m}  \\
		\tilde{\mathcal{S}}^{I} &= \sqrt{2}\psi_{m}I_{mn}a_{n}
	\end{split}
\end{equation} 
The $\mathfrak{su}(2)$ generators can be expressed:
\begin{equation}
	J_3 = \frac{1}{2}\left(\psi^{\dagger}_{m}\psi_{m}-1\right) \qquad J_+ = \frac{1}{2}I_{mn}\left(\psi^{\dagger}_{m}\psi^{\dagger}_{n}\right) = \psi^{\dagger}_2\psi^{\dagger}_1 \qquad J_- = J_{+}^{\dagger} = \psi_{1}\psi_{2}
\end{equation}
and the non-central R-symmetry:
\begin{equation}
	R^I = \frac{i}{2}I_{mn}\psi^{\dagger}_m\psi_n = \frac{i}{2}\left(-\psi^{\dagger}_1\psi_2 + \psi^{\dagger}_2 \psi_1\right)
\end{equation}
Switching to holomorphic coordinates  $z = X_1 + iX_2$ with respect to the canonical complex structure induces a basis change in the fermions and oscillator operators:
\begin{equation}
	\begin{split}
		\chi^{\dagger} &= \psi^{\dagger}_{1} + i\psi^{\dagger}_{2}  \\
		\tilde{\chi}^{\dagger} &= \psi^{\dagger}_{1} - i\psi^{\dagger}_{2}
	\end{split}
	\qquad\qquad
	\begin{split}
		\beta^{\dagger} &= a^{\dagger}_{1} + ia^{\dagger}_{2}  \\
		\tilde{\beta}^{\dagger} &= a^{\dagger}_{1} - ia^{\dagger}_{2}
	\end{split}
\end{equation} 
obeying:
\begin{equation}
	[\beta, \beta^{\dagger}] = [\tilde{\beta}, \tilde{\beta}^{\dagger}] = \{\chi, \chi^{\dagger}] = \{\tilde{\chi}, \tilde{\chi}^{\dagger}\} = 2  
\end{equation} 
To make the $\mathfrak{su}(2)$ R-symmetry rotating fermions manifest, define:
\begin{equation}
	\Psi^{A} = \frac{1}{\sqrt{2}}\left(\tilde{\chi}, \chi^{\dagger}\right) = \frac{1}{\sqrt{2}}\left(\psi_{1} + i\psi_{2}, \psi^{\dagger}_{1} + i\psi^{\dagger}_{2} \right) 
	= \left( \sqrt{2}\iota_{\frac{\partial}{\partial \bar{z}}}, \,\, \frac{1}{\sqrt{2}}dz\wedge \right)
\end{equation}
\begin{equation}
	\bar{\Psi}^{A} = \frac{1}{\sqrt{2}}\left(\psi^{\dagger}_{1} - i\psi^{\dagger}_{2}, \psi_{1} - i\psi_{2}\right)
	= \left(\frac{1}{\sqrt{2}}d\bar{z}\wedge , \,\,  \sqrt{2}\iota_{\frac{\partial}{\partial z}}\right)
\end{equation}
Here $\bar{\Psi}^{\bar{A}} = \left(\Psi^A\right)^{\dagger}$, and note that:
\begin{equation}
	\{ \Psi^{A}, \bar{\Psi}_{B} \}= \delta^A_B
\end{equation}
Notice that $\Psi^A$ forms an $SU(2)$ doublet:
\begin{equation}
	[J_3, \Psi^{1,2}]  = \mp \Psi^{1,2} \qquad [J_+, \Psi^1] = -i\Psi^2 \qquad [J_-, \Psi^2] = i\Psi^1
\end{equation} 
and  both $\Psi^A$ have R-charge $\frac{1}{2}$.\\

To construct the Hilbert Space, take a vacuum state $\ket{0}$ such that:
\begin{equation}
	\bar{\Psi}^A\ket{0} = \beta\ket{0} = \bar{\beta}\ket{0} = 0 \quad \Rightarrow \quad \ket{0} = d\bar{z} e^{- K}
\end{equation}
$\ket{0}$ has quantum numbers: $\{\Delta = 1, j = 0, d = -1, r = -1/2\}$. The vacuum state is a singlet under the $\mathfrak{su}(2)$. The action of $\beta^{\dagger}$ or $\bar{\beta}^{\dagger}$ raises the scaling dimension by 1, and the fermions commute with $\mathbb{L}_0$ and hence all states have $\Delta\geq1$.\\

We construct a basis of $E=0$ states. On flat space the homogeneous forms $z^k\bar{z}^l dz^{p}d\bar{z}^qe^{-K}$ form a basis of simultaneous eigenstates of the Cartan generators. The generator $\hat{D}^I = -i\mathcal{L}_{D^I}$ (we will henceforth abuse notation and call the operator on the Hilbert space and the Reeb vector the same thing) where on $\mathbb{C}$:
\begin{equation}
	D^{I} = i\left(z\frac{\partial}{\partial z} - \bar{z}\frac{\partial}{\partial\bar{z}}\right)
\end{equation}
and hence on the basis of homogeneous forms $\hat{D}^I$ has eigenvectors $p-q+k-l$. Similarly $\mathcal{H}$ has eigenvalues:
\begin{equation}
	E = \frac{1}{2}\left(\Delta + 2q -d_{\mathbb{C}}-(k-l)\right)
\end{equation}
Note that under $\Psi^1$ has $E$ charge -1 ($[E, \Psi^1] = -\Psi^1$), $\Psi^2$ is neutral, $\beta^{\dagger}$ is neutral, and $\bar{\beta}^{\dagger}$ has charge $1$. The vacuum state has $E$ charge 1 and $\Psi^1$ can be applied at most once, therefore the most general $E=0$ state is of the form:
\begin{equation}
	\ket{\phi} = \Psi^1 ({\Psi^2})^{p(2)} {\beta^{\dagger}}^n\ket{0} \qquad p(2)\in\{0,1\},\, n\in \mathbb{N}_0
\end{equation}
Geometrically this corresponds to a form:
\begin{equation}
	\ket{\phi} \propto \left(dz\right)^{p(2)} \left(z-2\frac{\partial}{\partial\bar{z}}\right)^{n}e^{-K}
\end{equation}
i.e. a holomorphic differential form on $\mathbb{C}$. So:
\begin{equation}
	\left(\mbox{States with E = 0}\right) = \mathbb{C}\left[z, dz\right]e^{-K}
\end{equation}

To compute the superconformal index note that the  fundamental variables are bosons $\beta, \beta^{\dagger}$ and fermions $\Psi^{1,2}$, which are all simultaneous eigenvalues of $\mathcal{H}, D^I, J_3+R^I$ and $F$. The trace separates into a product of contributions from each variable. The vacuum contribution is:
\begin{equation}\label{vacuumcontribution}
	\bra{0}(-1)^Fe^{-\beta \mathcal{H}}\tau^{D^{I}}\tilde{y}^{J_3+R^I}\ket{0} = e^{-\beta}\tau^{-1}\tilde{y}^{-\frac{1}{2}}
\end{equation}
The fermionic contributions are:
\begin{equation}
	\prod_A \mbox{tr}'_A\left[(-1)^Fe^{-\beta \mathcal{H}}\tau^{D^{I}}\tilde{y}^{J_3+R^I}\right] = \left(1-e^{\beta}\tau\right)(1-\tilde{y}\tau)
\end{equation}
where tr$'_{A}$ corresponds to the trace over $\langle \, \ket{0}, \Psi^A\ket{0}\, \rangle$ divided by the vacuum contribution (\ref{vacuumcontribution}).\\

To calculate the bosonic contribution note that geometrically:
\begin{equation}
	\beta^{\dagger} = \sqrt{\frac{1}{2}}\left(z-2\frac{\partial}{\partial \bar{z}}\right)\qquad
	\bar{\beta}^{\dagger} = \sqrt{\frac{1}{2}}\left(\bar{z}-2\frac{\partial}{\partial z}\right)
\end{equation}
So $\beta^{\dagger}$ raises $\Delta$ by 1, raises $d$ by 1 and leaves bidegree invariant, and $\bar{\beta}^{\dagger}$ raises $\Delta$ by 1, lowers $d$ by 1 and leaves bidegree invariant. Hence the bosonic contribution is:
\begin{equation}
	\left(\sum^{\infty}_{r=0} \tau^r\right)\left(\sum^{\infty}_{r=0} \tau^{-r}e^{-\beta r}\right) = \left(\frac{1}{1-\tau}\right)\left(\frac{1}{1-\tau^{-1}e^{-\beta}} \right)
\end{equation}
Multiplying all contributions together, we obtain the superconformal index for $\mathbb{C}$, which is indeed independent of $\beta$ as claimed:
\begin{equation}
	\mathcal{Z}_{\mathbb{C}}(\tau,\tilde{y}) = -\frac{(1-\tilde{y}\tau)}{\tilde{y}^{\frac{1}{2}}(1-\tau)}
\end{equation}
Notice this coincides with the form expected (\ref{scindexlocalisationformula}) from localisation with respect to $D^I$. To interpret this in terms of the representation content of the theory, compare with (\ref{scindexexpansion}):
\begin{equation}
	\mathcal{Z}_{\mathbb{C}}(\tau,\tilde{y}) =\sum_{r\in \mathbb{Z}/2}\tilde{y}^{r}I^{0r} + \sum_{d > 0, \,r\in \mathbb{Z}/2}\tau^{d}\,\tilde{y}^{r}\left(1-\frac{1}{\tilde{y}}\right)I^{dr} = \left({\tilde{y}^{\frac{1}{2}}\tau-\tilde{y}^{-\frac{1}{2}}}\right)\sum_{n=0}^{\infty}\tau^{n}
\end{equation}
and thus:
\begin{equation}
	I^{0,-\frac{1}{2}} = -1, \qquad I^{d,\frac{1}{2}} = 1 \quad(d>0), \qquad I^{d,r} = 0 \quad \forall \,\, \mbox{other}\,\,  \{d,r\}
\end{equation}
It is possible to show consistency by constructing the representations explicitly for $\mathbb{C}$. Using the geometric constraints that representations only exist with lowest weight $-\frac{1}{2}\leq j,r \leq \frac{1}{2}$, we have that:
\begin{equation}
	I^{0,-\frac{1}{2}} = \sum_{j=0,\frac{1}{2}} (-1)^{2j} N(S(j,d,-\frac{1}{2}+j)) = N(S(0,0,-1/2)) - N(S(1/2,0,0)) = -1
\end{equation}
\begin{equation}
	I^{d,\frac{1}{2}} = N\left(S\left(0,d,1/2\right)\right) = 1 \qquad d>0
\end{equation}
Note that  shown that $\Delta \geq 1$ for all states, thus we cannot have BPS representations with $j=d=0$, so:
\begin{equation}
	N(S(0,0,r)) = 0 \quad \forall r \quad \Rightarrow \quad N(S(1/2,0,0)) = 1
\end{equation}
We analyse both in turn:
\begin{itemize}
	\item $N\left(S\left(1/2,0,0\right)\right) =  N\left(S_{\frac{1}{2}}\left(1/2,0\right)\right)=1$: the lowest weight state is: $\Psi^1\ket{0} = \sqrt{2}e^{-K}$, which sits in the fundamental ($j = \frac{1}{2})$ representation of the $\mathfrak{su}(2)$: $\Psi^A\ket{0}$ and has $E=0$ as expected. It is also easy to check it is annihilated by all lowering operators in $\mathfrak{u}(1,1|2)$ as well as the geometric operators corresponding to $q^{1\pm}$, which are:
	\begin{equation}\label{q1pm}
		q^{1+} = \bar{\partial}^{\dagger}+ \iota_{\bar{z}\frac{\partial}{\partial\bar{z}}} \qquad\quad
		q^{1-} = \partial^{\dagger}+ \iota_{z\frac{\partial}{\partial z}}
	\end{equation}
	
	\item $N\left(S\left(0,d,1/2\right)\right) = N\left(S'_{\frac{1}{2}}\left(d,1/2\right)\right)=1$, for $d>0$: these multiplets have lowest weight states:
	\begin{equation}\label{lowestweightstate}
		\ket{\phi} = {\beta^{\dagger}}^{d-1}\Psi^1\Psi^2\ket{0} \propto {\beta^{\dagger}}^{d-1}\,dz\, e^{- K}
	\end{equation}
	And are $\mathfrak{su}(2)$ singlets. It is easy to check this is a lowest weight state using $L_- = -\frac{1}{2}\beta\tilde{\beta}$ and the geometric form of $J_-$ and $S, \tilde{S}, S^{I}$ and $\tilde{S}^{I}$. We can check this is annihilated by $q^{1+}$ and:
	\begin{equation}
		q^{2+} = - i \left(\partial-\partial K \wedge \right)
	\end{equation}
\end{itemize}
Therefore the index is in agreement with explicit constructions. It is easy now to compute the index for $\mathbb{C}^n$ by repeating the above with an extra index $I=1,...,n$ indicating each copy of $\mathbb{C}$. The index for $\mathbb{C}^n$ is then just the $n^{th}$ power of the index on $\mathbb{C}$ since the operators decompose into $\mathbb{C}$ blocks and the Hilbert space is the $n^{th}$ tensor power of that of $\mathbb{C}$. The vacuum, bosonic, fermionic contributions are all easily seen to individually be the $n^{th}$ power of those of $\mathbb{C}$, and the index is:
\begin{equation}
	\mathcal{Z}_{\mathbb{C}^n}(\tau,\tilde{y}) = \left(-\frac{(1-\tilde{y}\tau)}{\tilde{y}^{\frac{1}{2}}(1-\tau)}\right)^{n}
\end{equation}
and we can easily check agreement by constructing representations explicitly, which we neglect to do here in the interest of brevity.

\bibliographystyle{JHEP}
\bibliography{doreyzhang2019}

\end{document}